\newif\ifJOURNAL  
\JOURNALfalse
\newif\ifWP
\WPfalse
\newif\ifBASIC
\BASICfalse
\newif\ifFULL

\FULLfalse
\newif\ifLATIN
\LATINfalse

\WPtrue

\LATINtrue		
\ifWP\LATINtrue\fi	

\newif\ifnotJOURNAL	
\notJOURNALtrue
\ifJOURNAL\notJOURNALfalse\fi

\newif\ifnotFULL	
\notFULLtrue
\ifFULL\notFULLfalse\fi

\newif\ifnotLATIN	
\notLATINtrue
\ifLATIN\notLATINfalse\fi

\ifJOURNAL
  \newcommand*{\GTPI}{Vovk/Shafer:2008CAPM}
  \newcommand*{\GTPII}{GTP2}
  \newcommand*{\CTIV}{Vovk:2012FS-short}
  \newcommand*{\CTV}{Vovk:2011-local}
\fi
\ifWP
  \newcommand*{\GTPI}{GTP1}
  \newcommand*{\GTPII}{GTP2}
  \newcommand*{\CTIV}{GTP28arXiv}
  \newcommand*{\CTV}{GTP35arXiv}
\fi
\ifBASIC
  \newcommand*{\GTPI}{Vovk/Shafer:2008CAPM}
  \newcommand*{\GTPII}{GTP2}
  \newcommand*{\CTIV}{Vovk:2012FS} 
  \newcommand*{\CTV}{Vovk:2011-local} 
\fi

\ifJOURNAL
  \documentclass{lmj}
  \include{lmj.def}
  \LMJCLS%
  \usepackage{stmaryrd}
  \newcommand{\Extra}[1]{}
\fi

\ifWP
  \documentclass[10pt]{article}
\makeatletter

\newif\iftwodates
\twodatesfalse

\renewcommand\maketitle{\begin{titlepage}%
  \let\footnotesize\small
  \let\footnoterule\relax
  \let \footnote \thanks
  \null\vfil
  \vskip 30\p@
  \begin{center}%
    {\LARGE \bf \@title \par}%
    \vskip 3em%
    {\large
     \lineskip .75em%
     \begin{tabular}[t]{c}%
       \@author
     \end{tabular}\par}%
     \vskip 1.5em%
  \end{center}\par
  \vfill
  \begin{center}
    \raisebox{1.5cm}{\includegraphics[width=0.58\textwidth]%
      {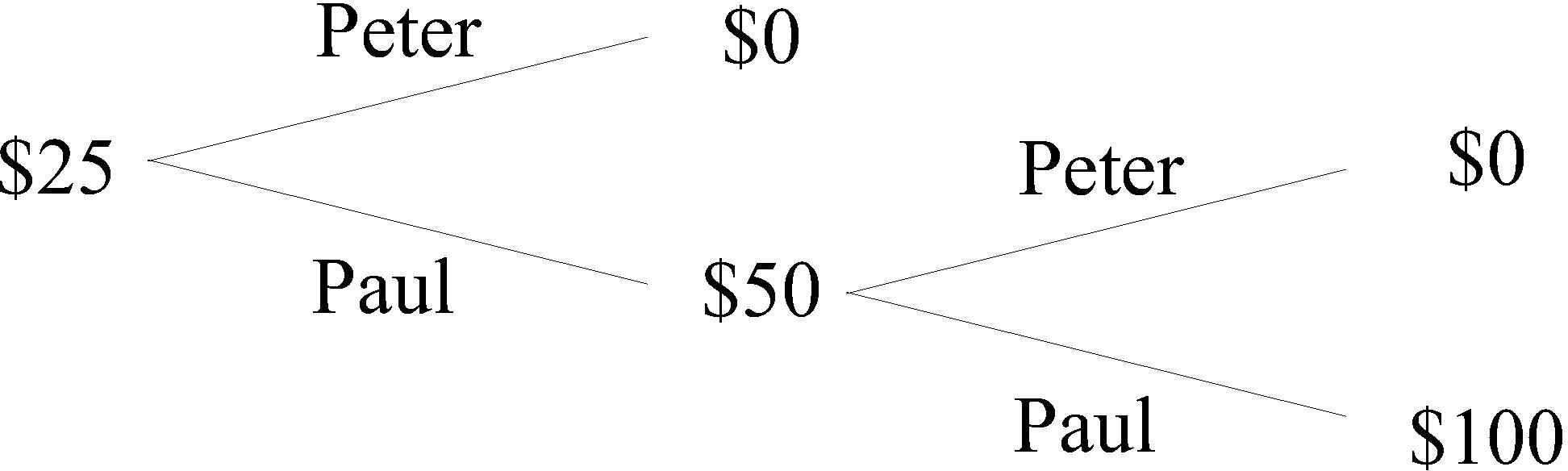}}%
    \hskip 3em%
    \includegraphics[width=0.29\textwidth]%
      {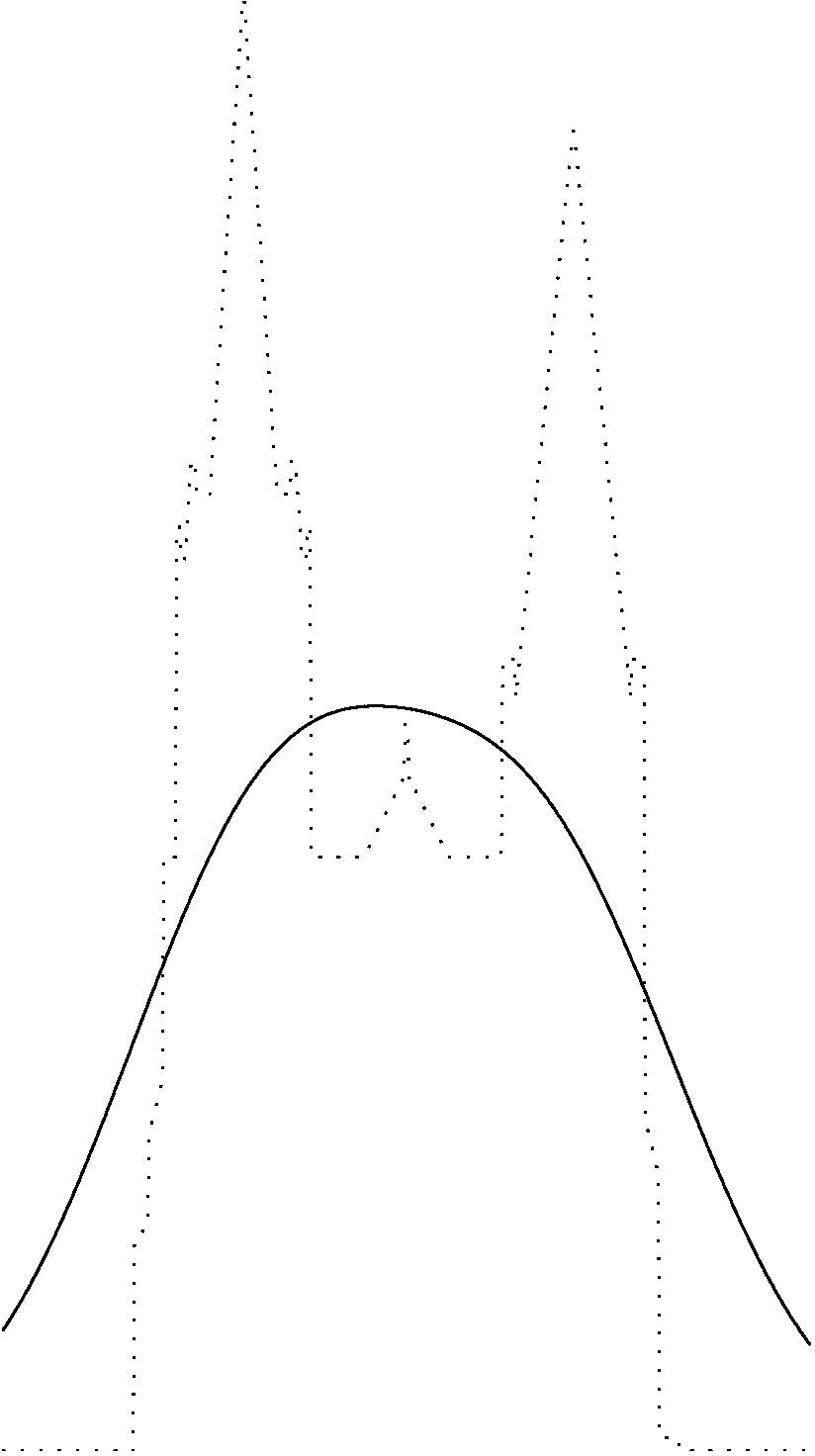}%
  \end{center}
  \@thanks
  \vfill
  \begin{center}
    {\large \bf The Game-Theoretic Probability and Finance Project}
  \end{center}
  \begin{center}
    {\large Working Paper \#\No}
  \end{center}
  \begin{center}
    {\iftwodates\large First posted \firstposted.
    Last revised \@date.\else\large\@date\fi}
  \end{center}
  \begin{center}
    Project web site:\\
    http://www.probabilityandfinance.com
  \end{center}
  \end{titlepage}%
  \setcounter{footnote}{0}%
  \global\let\thanks\relax
  \global\let\maketitle\relax
  \global\let\@thanks\@empty
  \global\let\@author\@empty
  \global\let\@date\@empty
  \global\let\@title\@empty
  \global\let\title\relax
  \global\let\author\relax
  \global\let\date\relax
  \global\let\and\relax
}

\renewenvironment{abstract}{%
  \titlepage
  \null\vfil
  \@beginparpenalty\@lowpenalty
  \begin{center}%
    \Large \bfseries \abstractname
    \@endparpenalty\@M
  \end{center}}%
  {\par\vfill\tableofcontents\endtitlepage}

\renewenvironment{thebibliography}[1]
  {\section*{\refname}%
  \addcontentsline{toc}{section}{\refname}
  \@mkboth{\MakeUppercase\refname}{\MakeUppercase\refname}%
  \list{\@biblabel{\@arabic\c@enumiv}}%
    {\settowidth\labelwidth{\@biblabel{#1}}%
    \leftmargin\labelwidth
    \advance\leftmargin\labelsep
    \@openbib@code
    \usecounter{enumiv}%
    \let\p@enumiv\@empty
    \renewcommand\theenumiv{\@arabic\c@enumiv}}%
    \sloppy
    \clubpenalty4000
    \@clubpenalty \clubpenalty
    \widowpenalty4000%
    \sfcode`\.\@m}
    {\def\@noitemerr
    {\@latex@warning{Empty `thebibliography' environment}}%
  \endlist}

\makeatother

  \usepackage{amsmath,amsthm,amsfonts,amssymb,latexsym,graphicx,stmaryrd}
  \newcommand{\Extra}[1]{}
\fi

\ifBASIC
  \documentclass[10pt]{article}
  \usepackage{amsmath,amsthm,amsfonts,amssymb,latexsym,stmaryrd}
  \newcommand{\Extra}[1]{}
\fi

\ifFULL
  \usepackage{color}
  \renewcommand{\Extra}[1]{\blue{#1}}
  
  \newcommand*{\blue}[1]{\textcolor{blue}{#1}}
  \newcommand*{\bluebegin}{\begingroup\color{blue}}
  \newcommand*{\blueend}{\endgroup}

\fi

\allowdisplaybreaks[4]

\ifnotLATIN
  \input{OT2enc.def}
  
  \usepackage{CJK}
\fi

\renewcommand{\theenumi}{(\alph{enumi})}

\newcommand*{\st}{\mathrel{|}}		

\newcommand*{\dd}{\mathrm{d}}		

\newcommand*{\K}{\mathcal{K}}		

\newcommand*{\BBB}{\mathcal{B}}		
\newcommand*{\FFF}{\mathcal{F}}		

\DeclareMathOperator{\III}{\boldsymbol{1}}	
\DeclareMathOperator{\argmin}{argmin}		
\DeclareMathOperator{\Ln}{Ln}			
\DeclareMathOperator{\mesh}{mesh}		
\DeclareMathOperator{\w}{w}			
\DeclareMathOperator{\osc}{osc}			
\DeclareMathOperator{\vi}{vi}			

\newcommand*{\lng}{_{\textrm{lg}}}   
\newcommand*{\shrt}{_{\textrm{sh}}}  
\newcommand*{\cont}{^{\textrm{c}}}   

\newcommand*{\bbbp}{\mathbb{P}}      

\DeclareMathOperator{\UpProb}{\overline{\bbbp}}    

\DeclareMathOperator{\MM}{M}          
\DeclareMathOperator{\DD}{D}          

\newcommand*{\bbbr}{\mathbb{R}}    
\newcommand*{\bbbd}{\mathbb{D}}    
\newcommand*{\bbbz}{\mathbb{Z}}    
\newcommand*{\bbbq}{\mathbb{Q}}    
\newcommand*{\bbbn}{\mathbb{N}}    

\ifnotJOURNAL 
  \theoremstyle{plain}
  \newtheorem{theorem}{Theorem} 
  \newtheorem{proposition}{Proposition}
  \newtheorem{corollary}{Corollary}
  \newtheorem{lemma}{Lemma}
  \theoremstyle{definition}
  \newtheorem{remark}{Remark}

  \makeatletter
  \@addtoreset{equation}{section}
  
  \makeatother
\fi

\ifWP
  \title{It\^o calculus without probability in idealized financial markets}
  \author{Vladimir Vovk}
  \newcommand*{\No}{36}
  \twodatestrue
  \newcommand*{\firstposted}{August 3, 2011}
\fi

\ifBASIC
  \title{It\^o calculus without probability in idealized financial markets}
  \author{Vladimir Vovk\\
  \texttt{v.vovk{\rm@}rhul.ac.uk}\\
  \texttt{http://vovk.net}}
\fi

\begin{document}

\ifnotJOURNAL
  \maketitle
\fi

\begin{abstract}
  We consider idealized financial markets
  in which price paths of the traded securities
  are c\`adl\`ag functions,
  imposing mild restrictions on the allowed size of jumps.
  We prove the existence of quadratic variation
  for typical price paths,
  where the qualification ``typical'' means that there is a trading strategy
  that risks only one monetary unit and brings infinite capital
  if quadratic variation does not exist.
  This result allows one to apply numerous known results
  in pathwise It\^o calculus to typical price paths;
  we give a brief overview of such results.
\end{abstract}

\section{Introduction}

It\^o calculus, based on the notion of the stochastic integral,
plays an important role in mathematical finance.
However, the usual construction of the stochastic integral
relies on statistical assumptions about security prices,
and it is not easy to come up with their realistic statistical models,
as witnessed by the proliferation of various competing models in recent years.
The pathwise stochastic integral and It\^o calculus
proposed in 1981 by F\"ollmer \cite{Follmer:1981}
and developed by numerous authors
(see, e.g., \cite{Norvaisa:2001,Cont/Fournie:2010})
do not depend on any statistical assumptions.
They, however, depend on the existence of quadratic variation for the allowed paths,
which is postulated.
The goal of this paper is to establish the existence of quadratic variation
of security prices under weak conditions
that can be justified from the economic point of view.

The existence of quadratic variation for typical continuous price paths
was established in \cite{\CTIV},
where it served as a tool for studying properties of typical price paths,
such as their volatility.
This paper proves the existence of quadratic variation
under a weaker assumption than in \cite{\CTIV}:
the price paths are assumed to be c\`adl\`ag without huge jumps.
A related result was obtained in \cite{\CTV}:
that paper shows that the $p$-variation of typical prices paths
is finite for $p>2$.
The assumptions of \cite{\CTV} are not comparable
to the assumptions of this paper:
on one hand, there are no restrictions on the size of jumps in \cite{\CTV},
but on the other hand,
the price paths are assumed nonnegative
(albeit strong, this assumption is natural
in the context of financial markets).
In this paper, we will often be using results and methods of \cite{\CTV}.

This paper, like \cite{\CTIV} and \cite{\CTV},
is written in the tradition of game-theoretic probability
(see, e.g., \cite{Shafer/Vovk:2001,%
Takeuchi/etal:2009,%
Miyabe/Takemura:2012,Miyabe/Takemura:2013,Miyabe/Takemura:2014}).
In game-theoretic probability,
probability-like notions (such as the upper probabilities of events and upper prices of functions)
are defined in terms of
idealized financial markets,
and mathematical finance is a natural field of application.
The key technical tool used in this paper will be ``high-frequency limit order strategies'',
introduced in game-theoretic probability by \cite{Takeuchi/etal:2009}.

We start the main part of the paper by defining in Section \ref{sec:QV-def}
the notion of quadratic variation,
a version of F\"ollmer's pathwise definition 
adapted to our goals.
The next section, Section \ref{sec:definitions},
states and discusses our assumption about the jumps of price paths
and defines the notion of a typical price path.
Our main result, the existence of quadratic variation for typical price paths,
is established in Section~\ref{sec:QV-proof}.
Section \ref{sec:multidim} extends this result to typical vector-valued price paths;
in particular, it proves the existence of quadratic covariation between two price paths.
Section \ref{sec:other-definitions} compares
the notion of quadratic variation used in this paper
with Norvai\u{s}a's \cite{Norvaisa:2001} and F\"ollmer's \cite{Follmer:1981};
we show that under natural conditions
(satisfied in the main example in this paper)
the first two notions of quadratic variation are equivalent,
whereas F\"ollmer's notion becomes equivalent to them
when another natural condition is added.
Section~\ref{sec:implications} summarizes some of the known results
for price paths possessing quadratic variation.

\subsection*{Notation}

The set of all real (resp.\ rational, resp.\ integer) numbers
is denoted $\bbbr$ (resp.\ $\bbbq$, resp.\ $\bbbz$).
The set of all natural numbers is denoted $\bbbn$: $\bbbn:=\{1,2,\ldots\}$;
set $\bbbn_0:=\bbbn\cup\{0\}$.
We use the usual notation $u\vee v:=\max(u,v)$, $u\wedge v:=\min(u,v)$,
and $u^+:=u\vee0$.
If $s<t$, we will use the notation $s\vee u\wedge t$
to mean $(s\vee u)\wedge t=s\vee(u\wedge t)$.
The expression $\inf\emptyset$ is always interpreted as $\infty$.

The set of all c\`adl\`ag
(i.e., continuous on the right and having limits on the left)
functions $\omega:[0,T]\to\bbbr$,
where $T>0$,
will be denoted $D[0,T]$.
If $\omega\in D[0,T]$ and $t\in(0,T]$,
we set $\omega(t-):=\lim_{s\uparrow t}\omega(s)$
and $\Delta\omega(t):=\omega(t)-\omega(t-)$.
As usual,
$C^n(\bbbr)$ stands for the set of all functions $f:\bbbr\to\bbbr$
that are $n$ times continuously differentiable.

\section{Pathwise quadratic variation}
\label{sec:QV-def}

Let $\omega:[0,T]\to\bbbr$ be a c\`adl\`ag function,
interpreted as the price path of a financial security
over the time period $[0,T]$ whose end-point $T>0$ is fixed throughout the paper.
In this section we give a modification of F\"ollmer's \cite{Follmer:1981} definition
of the quadratic variation of $\omega$;
F\"ollmer's definition itself will be discussed in Section~\ref{sec:other-definitions}.

A \emph{partition} (of $[0,T]$) is a finite sequence of numbers
$0=t_0<t_1<\cdots<t_m\le T$;
we also set $t_k:=\infty$ for $k>m$.
The \emph{mesh} of this partition is
$\max_{k\in\bbbn}\left|(t_{k}\wedge T) - (t_{k-1}\wedge T)\right|$.

Let $\pi=(\pi^0,\pi^1,\pi^2,\ldots)$
be a nested sequence of partitions:
for each $n\in\bbbn$, each element of $\pi^{n-1}$ is an element of $\pi^n$
(in this paper we concentrate on nested sequences of partitions).
The $n$th approximation, $n=0,1,2,\ldots$,
to the quadratic variation of $\omega$
along $\pi$ is defined by
\begin{equation}\label{eq:A1}
  A^{n,\pi}_t
  :=
  \sum_{k=1}^{\infty}
  \left(
    \omega(\pi^n_k\wedge t)
    -
    \omega(\pi^n_{k-1}\wedge t)
  \right)^2,
  \quad
  t\in[0,T],
\end{equation}
where, for each $n\in\bbbn_0$,
$\pi^n_k$, $k=0,1,\ldots$, are the elements of $\pi^n$:
$\pi^n=(\pi^n_0,\pi^n_1,\pi^n_2,\ldots)$
and there is $m=m(n)\in\bbbn_0$ such that $0=\pi^n_0<\pi^n_1<\cdots<\pi^n_m\le T$
and $\pi^n_k=\infty$ for all $k>m$.
By $A^{n,\pi}$ we will mean the function $t\in[0,T]\mapsto A^{n,\pi}_t$ in $D[0,T]$.

In this paper we will be interested in the uniform metric on $D[0,T]$:
$$
  \rho(\omega,\omega')
  :=
  \sup_{t\in[0,T]}
  \left|
    \omega(t) - \omega'(t)
  \right|.
$$
(This is an unusual metric for $D[0,T]$,
but standard for the set $C[0,T]$ of continuous functions on $[0,T]$.)
Notice that $D[0,T]$ is complete in the metric $\rho$;
this will be used in the proof of Theorem~\ref{thm:1D} below.


We say that $\omega$ \emph{has quadratic variation along $\pi$},
where $\pi$ is a nested sequence of partitions,
if the sequence $A^{n,\pi}$ converges in the uniform metric $\rho$.
The limit, when it exists, is denoted $A^{\pi}$
and called the \emph{quadratic variation of $\omega$ along $\pi$}.

A sequence $\pi$ of partitions is \emph{dense}
if $\lim_{n\to\infty}\mesh(\pi^n)=0$,
where $\mesh(\pi^n)$ is the mesh of $\pi^n$.
Quadratic variation along a sequence of partitions is usually
(see, e.g., \cite{Follmer:1981,Norvaisa:2001})
defined only for dense sequences of partitions.
The sequences of partitions considered in this paper are not always dense,
and instead we will use the following property.
For $\omega\in D[0,T]$,
we say that a nested sequence $\pi$ of partitions $\pi^n$,
$n=0,1,2,\ldots$,
\emph{exhausts} $\omega$ if:
\begin{itemize}
\item
  each $t\in(0,T]$ such that $\Delta\omega(t)\ne0$
  belongs to $\pi^n$ for some $n$
  (equivalently, from some $n$ on);
\item
  each open interval $(u,v)$ in which $\omega$ is not constant
  contains at least one element of $\pi^n$ for some $n$
  (equivalently, from some $n$ on).
\end{itemize}

For $\omega\in D[0,T]$ and an interval $I\subseteq[0,T]$ set
$\w_{\omega}I:=\sup_{s_1,s_2\in I}\lvert\omega(s_2)-\omega(s_1)\rvert$.
The \emph{oscillation} of $\omega\in D[0,T]$ over a partition $\pi^n$ is defined as
$$
  \osc_{\pi^n}(\omega)
  :=
  \max_{k\in\bbbn}
  \w_{\omega}
  \left[
    \pi^n_{k-1}\wedge T, \pi^n_{k}\wedge T
  \right)
  =
  \max_{k\in\bbbn}
  \w_{\omega}
  \left(
    \pi^n_{k-1}\wedge T, \pi^n_{k}\wedge T
  \right).
$$

\begin{lemma}\label{lem:osc}
  If a nested sequence of partitions $\pi$ exhausts $\omega\in D[0,T]$,
  \begin{equation}\label{eq:osc}
    \lim_{n\to\infty}
    \osc_{\pi^n}(\omega)
    =
    0.
  \end{equation}
\end{lemma}

\ifJOURNAL
  \begin{pf}
\fi
\ifnotJOURNAL
  \begin{proof}
\fi
  Let $\epsilon>0$.
  There exist points $0=t_0<t_1<\cdots<t_r=T$
  such that $\w_{\omega}[t_{i-1},t_i)<\epsilon$ for all $i\in\{1,\ldots,r\}$
  (\cite{Billingsley:1968}, Lemma 14.1).
  Let $n$ be so large that $\pi^n$
  contains all $t_i$, $i\in\{1,\ldots,r\}$, with $\Delta\omega(t_i)\ne0$
  and contains a point in each interval $(t_{i-1},t_i)$, $i\in\{1,\ldots,r\}$,
  such that $\omega$ is not constant in $(t_{i-1},t_i)$.
  For any $k\in\bbbn_0$, the interval $(\pi^n_k\wedge T,\pi^n_{k+1}\wedge T)$
  does not contain any points $t_i$ with $\Delta\omega(t_i)\ne0$
  and does not contain any intervals $(t_{i-1},t_i)$, $i\in\{1,\ldots,r\}$,
  where $\omega$ is not constant.
  Therefore, $\w_{\omega}(\pi^n_k\wedge T,\pi^n_{k+1}\wedge T)<2\epsilon$.
  Since $\epsilon$ can be arbitrarily small,
  this completes the proof of (\ref{eq:osc}).
\ifJOURNAL
  \qed
  \end{pf}
\fi
\ifnotJOURNAL
  \end{proof}
\fi

\begin{lemma}\label{lem:increasing}
  Suppose $\omega\in D[0,T]$
  and $\pi$ is a nested sequence of partitions that exhausts $\omega$.
  If the quadratic variation $A^{\pi}$ of $\omega$ along $\pi$ exists,
  it is a nondecreasing c\`adl\`ag function satisfying $A^{\pi}_0=0$
  with jumps $\Delta A^{\pi}_t=(\Delta\omega(t))^2$ for all $t\in(0,T]$.
\end{lemma}

\ifJOURNAL
  \begin{pf}
\fi
\ifnotJOURNAL
  \begin{proof}
\fi
  The equality $A^{\pi}_0=0$ is obvious,
  and $A^{\pi}$ is c\`adl\`ag as the uniform limit of c\`adl\`ag functions.
  (These statements do not rely on $\pi$ exhausting $\omega$,
  but the other statements of the lemma
  do require a density condition for $\pi$:
  they fail, for example, when $\pi^0=\pi^1=\cdots=(0,\infty,\infty,\ldots)$.)

  Let us now prove that $A^{\pi}$ is nondecreasing.
  If it is not,
  there exist $t_1<t_2\le T$ and $\epsilon\in(0,1)$ such that
  $A^{\pi}_{t_1} > A^{\pi}_{t_2} + 2\epsilon$.
  For sufficiently large $n$, we will have
  \begin{equation}\label{eq:strange}
    A^{n,\pi}_{t_1} > A^{n,\pi}_{t_2} + \epsilon.
  \end{equation}
  Let $n$ be so large that, in addition, $\osc_{\pi^n}(\omega)<\epsilon$
  (cf.\ Lemma~\ref{lem:osc}).
  Let $t_1\in[\pi^n_{k'-1},\pi^n_{k'})$ and $t_2\in[\pi^n_{k''},\pi^n_{k''+1})$;
  notice that $k''\ge k'$.
  Since $A^{n,\pi}_{t_2}\ge A^{n,\pi}_{\pi^n_{k''}}$
  and $A^{n,\pi}_{\pi^n_{k+1}}\ge A^{n,\pi}_{\pi^n_{k}}$
  for $k=k',\ldots,k''-1$,
  (\ref{eq:strange}) implies
  $
    A^{n,\pi}_{t_1} > A^{n,\pi}_{\pi^n_{k'}} + \epsilon
  $.
  This, however, contradicts
  \begin{multline*} 
    A^{n,\pi}_{t_1}
    -
    A^{n,\pi}_{\pi^n_{k'}}
    =
    \left(
      \omega(t_1)
      -
      \omega(\pi^n_{k'-1})
    \right)^2
    -
    \left(
      \omega(\pi^n_{k'})
      -
      \omega(\pi^n_{k'-1})
    \right)^2
    \\   
    \le
    \left(
      \omega(t_1)
      -
      \omega(\pi^n_{k'-1})
    \right)^2
    <
    \epsilon^2.
  \end{multline*} 

  It remains to prove that $\Delta A^{\pi}_t=(\Delta\omega(t))^2$ for all $t\in(0,T]$.
  Fix $t\in(0,T]$ and $\epsilon>0$.
  Let $n$ be so large that $\osc_{\pi^n}(\omega)<\epsilon$
  (cf.\ Lemma~\ref{lem:osc}).
  Define $k\in\bbbn_0$ by the condition $t\in(\pi^n_{k},\pi^n_{k+1}]$.
  Since
  \begin{multline*}  
    \Delta A^{n,\pi}_t
    =
    \left(
      \omega(t) - \omega(\pi^n_{k})
    \right)^2
    -
    \left(
      \omega(t-) - \omega(\pi^n_{k})
    \right)^2
    \\  
    =
    (\Delta\omega(t))^2
    +
    2\Delta\omega(t)
    \left(
      \omega(t-) - \omega(\pi^n_{k})
    \right),
  \end{multline*} 
  we have
  $
    \left|\Delta A^{n,\pi}_t-(\Delta\omega(t))^2\right|
    \le
    2\epsilon\lvert\Delta\omega(t)\rvert
  $.
  As $\Delta A^{n,\pi}_t\to\Delta A^{\pi}_t$
  and $\epsilon$ can be arbitrarily small,
  we obtain $\Delta A^{\pi}_t=(\Delta\omega(t))^2$.
\ifJOURNAL
  \qed
  \end{pf}
\fi
\ifnotJOURNAL
  \end{proof}
\fi

\section{Typical price paths}
\label{sec:definitions}

We consider a perfect-information game between two players
called Reality (financial market) and Sceptic (speculator).
Reality outputs a c\`adl\`ag function $\omega:[0,T]\to\bbbr$,
interpreted as the price path of a financial security,
and Sceptic tries to profit by trading in $\omega$.
First Sceptic presents his trading strategy and then Reality chooses $\omega$.
In the first two subsections of this section
we will formalize this picture
(often following \cite{\CTV})
by defining the allowed moves for Reality and strategies for Sceptic.

\subsection*{Sample space}

Let $\psi:[0,\infty)\to(0,\infty)$ be a nondecreasing function,
fixed through most of this paper.
The set of allowed moves for Reality (our \emph{sample space}) is
\begin{equation}\label{eq:moderate}
  \Omega_{\psi}
  :=
  \left\{
    \omega\in D[0,T]
    \bigm|
    \forall t\in(0,T]:
    \left|
      \Delta\omega(t)
    \right|
    \le
    \psi
    \left(
      \sup_{s\in[0,t)}
      \lvert\omega(s)\rvert
    \right)
  \right\}.
\end{equation}
The function $\psi$ determines the allowed size of the jumps.
It can be arbitrarily large (but needs, however, to be known in advance).
Our conclusions (e.g., in Theorems \ref{thm:1D} and \ref{thm:multidim})
will not depend on $\psi$;
therefore, our results become stronger as $\psi$ becomes larger.

Two natural examples of restrictions on the jumps are:
\begin{itemize}
\item
  The absolute values $\lvert\Delta\omega(t)\rvert$ of $\omega$'s jumps
  never exceed a known constant $c>0$.
  Such price paths $\omega$ belong to $\Omega_c$.
\item
  The price path $\omega$ is known to be nonnegative,
  and the relative values
  $\Delta\omega(t)/\sup_{s\in[0,t)}\lvert\omega(s)\rvert$
  (with $0/0:=0$)
  of $\omega$'s jumps w.r.\ to their largest value so far
  never exceed a known constant $c>0$.
  Such $\omega$ belong to $\Omega_{\psi}$, where $\psi(u)=(c\vee1)u$.
  (There is no need to explicitly restrict downward jumps when $\omega\ge0$:
  they are restricted automatically by the current value of the security.)
\end{itemize}
The first example is the simplest one mathematically
and will be used in the proof of our main result in Section~\ref{sec:QV-proof}.
The second example is more relevant to many real financial markets,
and its generalized and elaborated version will be discussed at the end of this section.


\subsection*{Trading strategies}

For each $t\in[0,T]$,
$\FFF^{\circ}_t$ is defined to be the smallest $\sigma$-algebra on $\Omega_{\psi}$
that makes all functions
$\omega\mapsto\omega(s)$, $s\in[0,t]$,
measurable;
$\FFF_t$ is defined as the universal completion of $\FFF^{\circ}_t$.
A \emph{process} (more fully, adapted process) $S$ is a family of functions
$S_t:\Omega_{\psi}\to[-\infty,\infty]$, $t\in[0,T]$,
each $S_t$ being $\FFF_t$-measurable.
An \emph{event} is an element of the $\sigma$-algebra $\FFF_T$.
Stopping times $\tau:\Omega_{\psi}\to[0,T]\cup\{\infty\}$
w.r.\ to the filtration $(\FFF_t)$
and the corresponding $\sigma$-algebras $\FFF_{\tau}$
are defined as usual;
$\omega(\tau(\omega))$ and $S_{\tau(\omega)}(\omega)$
will often be simplified to $\omega(\tau)$ and $S_{\tau}(\omega)$,
respectively.

The class of allowed strategies for Sceptic is defined in two steps.
A \emph{simple trading strategy} $G$ consists of:
(a) a nondecreasing infinite sequence of stopping times
$\tau_1\le\tau_2\le\cdots$ such that,
for each $\omega\in\Omega_{\psi}$,
$\tau_n(\omega)<\infty$ for only finitely many $n$;
(b) for each $n=1,2,\ldots$,
a bounded $\FFF_{\tau_{n}}$-measurable function $h_n$.
To such $G$ and an \emph{initial capital} $\alpha\in\bbbr$
corresponds the \emph{simple capital process}
\begin{equation}\label{eq:simple-capital}
  \K^{G,\alpha}_t(\omega)
  :=
  \alpha
  +
  \sum_{n=1}^{\infty}
  h_n(\omega)
  \bigl(
    \omega(\tau_{n+1}\wedge t)-\omega(\tau_n\wedge t)
  \bigr),
  \quad
  t\in[0,T];
\end{equation}
the value $h_n(\omega)$ will be called the \emph{position}
taken at time $\tau_n$,
and $\K^{G,\alpha}_t(\omega)$ will sometimes be referred to
as Sceptic's capital at time $t$.
Notice that the sum of finitely many simple capital processes
is again a simple capital process.

A \emph{nonnegative capital process} is any process $S$
that can be represented in the form
\begin{equation}\label{eq:nonnegative-capital}
  S_t(\omega)
  :=
  \sum_{n=1}^{\infty}
  \K^{G_n,\alpha_n}_t(\omega),
\end{equation}
where the simple capital processes $\K^{G_n,\alpha_n}_t(\omega)$
are required to be nonnegative, for all $t$ and $\omega$,
and the nonnegative series $\sum_{n=1}^{\infty}\alpha_n$
is required to converge in $\bbbr$
(intuitively,
the total capital invested has to be finite).
The sum (\ref{eq:nonnegative-capital}) is always nonnegative,
but we allow it to take value $+\infty$.
Since $\K^{G_n,\alpha_n}_0(\omega)=\alpha_n$ does not depend on $\omega$,
$S_0(\omega)$ also does not depend on $\omega$
and will sometimes be abbreviated to $S_0$.

\subsection*{Upper price}

The \emph{upper price} of a set $E\subseteq\Omega_{\psi}$
is defined as
\begin{equation}\label{eq:upper-probability}
  \UpProb(E)
  :=
  \inf
  \bigl\{
    S_0
    \bigm|
    \forall\omega\in\Omega_{\psi}:
    S_T(\omega)
    \ge
    \III_E(\omega)
  \bigr\},
\end{equation}
where $S$ ranges over the nonnegative capital processes
and $\III_E$ stands for the indicator of $E$.
Notice that $\UpProb(\Omega_{\psi})=1$
(in the terminology of \cite{Shafer/Vovk:2001},
our game protocol is ``coherent''):
indeed, $\UpProb(\Omega_{\psi})<1$ would mean that some nonnegative capital process
increases between time $0$ and $T$ for all $\omega\in\Omega_{\psi}$,
and this is clearly impossible for constant $\omega$.

\begin{remark}
  The notion of upper price (introduced in \cite{Vovk:1993logic})
  is similar to the notion of upper hedging price
  (see, e.g., \cite{Follmer/Schied:2011}, Section 7.3)
  except that the latter replaces the requirement
  $
    \forall\omega:
    S_T(\omega)
    \ge
    \III_E(\omega)
  $
  in (\ref{eq:upper-probability}) by ``$S_T\ge\III_E$ a.s.'',
  thus requiring a statistical model for prices
  (which is, however, only used via its family of events of probability zero).
\end{remark}

\begin{remark}
  An alternative probability-free definition of upper price
  is given in \cite{Perkowski/Promel:2013}, Definition~1,
  and \cite{Perkowski/Promel:2014}, Definition~3.1.
  See \cite{Perkowski/Promel:2013}, Section~2.1, for a comparison of the two definitions;
  an advantage of our definition for the purpose of this paper
  is that our mathematical results
  are at least as strong (and possibly stronger) when they use our definition
  than when they use the definition of \cite{Perkowski/Promel:2013,Perkowski/Promel:2014}
  (see \cite{Perkowski/Promel:2014}, Lemma~10).
\end{remark}

We say that $E\subseteq\Omega_{\psi}$ is \emph{null} if $\UpProb(E)=0$.
A property of $\omega\in\Omega_{\psi}$ will be said to hold
\emph{for typical $\omega$} if the set of $\omega$ where it fails is null.
Correspondingly,
a set $E\subseteq\Omega_{\psi}$ is \emph{almost certain}
if $\UpProb(\Omega_{\psi}\setminus E)=0$.
All null sets $E$ are automatically \emph{strictly null}:
there exists a nonnegative capital process $S$ with $S_0=1$
such that $S_T(\omega)=\infty$ for all $\omega\in E$
(indeed, if $E$ is null, we can sum nonnegative capital processes $S^n$, $n\in\bbbn$,
such that $S^n_0=2^{-n}$ and $S^n_T\ge\III_E$).
Notice that the union of countably (in particular, finitely) many null sets
is also a null set.

\subsection*{Our definitions in view of margin requirements}

In the rest of this section we discuss our definitions
in view of the margin requirements that traders usually have to comply with.
Suppose, for concreteness, that $\omega$ is the price path of a common stock.
We then have $\omega(t)\ge0$ for all $t\in[0,T]$.
Our definition (\ref{eq:simple-capital}) implicitly assumes the following picture.
Sceptic starts from the amount $\alpha$ in his margin account,
and he never adds funds to or withdraws funds from the account.
He is allowed to take both positive and negative positions in the stock
(can go both long and short),
but the capital in his account should always stay nonnegative.
Let us consider a somewhat more realistic picture where Sceptic is required
not only to keep the capital in the margin account nonnegative
but also to satisfy margin requirements.

In general, different margin requirements apply to long and short positions.
The rules can be summarized as follows.
In the case of a long position,
the capital in the margin account should not only be nonnegative
but should stay nonnegative in the imaginary event
that the stock price immediately drops by $100c\lng$\%,
where $c\lng$ is the minimum margin requirement for long positions.
In the case of a short position,
the capital in the margin account should not only be nonnegative
but should stay nonnegative in the imaginary event
that the stock price immediately rises by $100c\shrt$\%,
where $c\shrt$ is the minimum margin requirement for short positions.
Different values may be used for $c\lng$ and $c\shrt$
at the time when the position is opened and at later times;
they are called the initial and maintenance margin requirements, respectively.
The current initial margin requirements stipulated by the Federal Reserve Board
\cite{FRB:220.12} are the same for both long and short positions: $c\lng=c\shrt=0.5$.
The maintenance margin requirement never exceeds the initial margin requirement,
and we make Sceptic's task harder by setting the former to the latter.
As $\omega\ge0$, we always assume $c\lng\le1$.

If the margin requirement becomes violated, the trader receives
a request, known as a margin call, to add funds to the account.
If the margin call is ignored, the account can be liquidated.
We will assume that the margin account provides ``non-recourse'' loans
on the part of the broker,
so that the trader is not responsible for any possible shortfall after liquidation.

\begin{remark}
  Although being non-recourse is a common feature of some related kinds of loan,
  such as stock loans,
  margin loans themselves are legally recourse loans in the USA.
  However, the assumption that margin loans are non-recourse is not unusual
  (\cite{Fortune:2000}, p.~29) and there is a view that ``margin loans,
  while legally recourse loans, might be in a limbo,
  somewhere between recourse and nonrecourse''
  (\cite{Fortune:2000}, p.~38).
\end{remark}

The trading strategy developed in this paper (see Theorem~\ref{thm:1D})
starts with one monetary unit,
makes sure that the capital is always nonnegative,
and brings infinite capital if the quadratic variation of $\omega$ does not exist;
it, however, assumes that $\omega\in\Omega_{\psi}$.
(What it means for the quadratic variation of $\omega\in\Omega_{\psi}$ to exist
will be defined in the next section and is not important for the current discussion.)
The strategy, however, has two disadvantages:
\begin{itemize}
\item
  whereas ensuring that its capital is always nonnegative when $\omega\in\Omega_{\psi}$,
  it is not guaranteed to satisfy the margin requirements;
\item
  the strategy can lead to a negative capital when $\omega\notin\Omega_{\psi}$.
\end{itemize}
If we, however, apply our result to the sample space $\Omega_{\psi'}$
in place of $\Omega_{\psi}$,
where
\begin{equation}\label{eq:psi-prime}
  \psi'(u)
  :=
  1
  \vee
  \bigl(
    (1+c\shrt)\psi(u)
    +
    c\shrt u
  \bigr),
\end{equation}
we will obtain a trading strategy satisfying the following stronger guarantees:
it still starts with one monetary unit;
it makes sure that the margin requirements are satisfied if $\omega\in\Omega_{\psi}$;
it brings infinite capital if $\omega\in\Omega_{\psi}$
but the quadratic variation of $\omega$ does not exist;
it never loses more than the one monetary unit
(because the strategy always ignores margin calls).

Let us check that the trading strategy constructed for $\Omega_{\psi'}$
will indeed satisfy the margin requirements for each of the constituent accounts
(corresponding to the addends in (\ref{eq:nonnegative-capital})),
supposing $\omega\in\Omega_{\psi}$.
Since $\psi'\ge1$, at each time when the position is long
the capital in the account will be no less
than the value of the stock in the account at this time
(otherwise, a 100\% downward jump in the price of $\omega$
would have led to a negative capital).
Since $c\lng\le1$, the margin requirement will be satisfied.
It remains to consider a time $t$ when the position is short.
The worst case is when the price jumps up by the largest allowed amount
becoming $\omega(t)+\psi(\sup_{s\in[0,t)}\omega(s))$.
The margin requirements will be satisfied
if the capital resulting from the imaginary event
that the price again jumps up by a factor of $1+c\shrt$
is still nonnegative.
This is guaranteed by our strategy since, by (\ref{eq:psi-prime}),
$$
  (1+c\shrt)
  \left(
    \omega(t)
    +
    \psi
    \left(
      \sup_{s\in[0,t)}\omega(s)
    \right)
  \right)
  \le
  \omega(t)
  +
  \psi'
  \left(
    \sup_{s\in[0,t)}\omega(s)
  \right).
$$


\section{Existence of quadratic variation}
\label{sec:QV-proof}

In this section we define a suitable sequence of partitions $\tau^n(\omega)$
for each $\omega\in\Omega_{\psi}$
and show that the quadratic variation of $\omega$ along this sequence
exists for typical~$\omega$.

For each $n\in\bbbn_0$,
let $\bbbd^n:=\{k2^{-n}\st k\in\bbbz\}$
and define a sequence of stopping times $\tau^n_k$, $k=0,1,2,\ldots$,
and a sequence $D^n_k$, $k=0,1,2,\ldots$,
of $\FFF_{\tau^n_k}$-measurable functions inductively by
$\tau^n_{0}:=0$, $D^n_0:=\sup(\bbbd^n\cap(-\infty,\omega(0)])$,
\begin{align}
  \tau^n_k(\omega)
  &:=
  \inf
  \left\{
    t\in[\tau^n_{k-1}(\omega),T]
    \st
    \llbracket\omega(\tau^n_{k-1}),\omega(t)\rrbracket
    \cap
    (\bbbd^n\setminus\{D^n_{k-1}(\omega)\})
    \ne
    \emptyset
  \right\},\notag\\
  D^n_k(\omega)
  &\in
  \argmin_
  {
    D
    \in
    \llbracket
      \omega(\tau^n_{k-1}),\omega(\tau^n_k)
    \rrbracket
    \cap
    \left(
      \bbbd^n\setminus\{D^n_{k-1}(\omega)\}
    \right)
  }
  \left|
    D - \omega(\tau^n_k)
  \right|
  \label{eq:D}
\end{align}
for $k=1,2,\ldots$,
where we use the notation
\begin{equation*}
  \llbracket u,v\rrbracket
  :=
  \begin{cases}
    [u,v] & \text{if $u\le v$}\\
    [v,u] & \text{if $u>v$}
  \end{cases}
\end{equation*}
for the convex closure of the set $\{u,v\}$.
Notice that the $\argmin$ in (\ref{eq:D}) is a one-element set,
and so $D^n_k(\omega)$ is determined uniquely.

We will check that $\tau^n_k$ are indeed stopping times
and that $D^n_k$ are $\FFF_{\tau^n_k}$-measurable
in Lemma~\ref{lem:tau} below.
If $\tau^n_k(\omega)=\infty$
(arising from our convention $\inf\emptyset:=\infty$),
we set, e.g., $D^n_k(\omega):=0$.
The intuition behind $D^n_k$ is that it is the current element of $\bbbd^n$
at time $\tau^n_k$;
when $\omega$ is continuous, $D^n_k=\omega(\tau^n_k)$
(this is the case considered in \cite{\CTIV}),
but in general we only have $\left|D^n_k-\omega(\tau^n_k)\right|<2^{-n}$,
assuming $\tau^n_k<\infty$.
Let $\tau^n(\omega)$ be the partition
$0=\tau^n_0(\omega)\le\tau^n_1(\omega)\le\tau^n_2(\omega)\le\cdots$
(in fact, we will have $\tau^n_k(\omega)<\tau^n_{k+1}(\omega)$
unless $\tau^n_k(\omega)=\infty$),
and let $\tau(\omega)$ be the sequence of the partitions $\tau^n(\omega)$.

\begin{theorem}\label{thm:1D}
  Let $\psi:[0,\infty)\to(0,\infty)$ be a nondecreasing function.
  Typical $\omega\in\Omega_{\psi}$
  have quadratic variation along $\tau(\omega)$.
\end{theorem}

The theorem says that for typical $\omega\in\Omega_{\psi}$
the sequence of functions
\begin{equation}\label{eq:A2}
  A^n_t(\omega)
  :=
  \sum_{k=1}^{\infty}
  \left(
    \omega(\tau^n_k\wedge t)
    -
    \omega(\tau^n_{k-1}\wedge t)
  \right)^2,
  \quad
  n=0,1,2,\ldots,
\end{equation}
converges in the uniform metric to a function $A(\omega)\in D[0,T]$.
(In this section, we omit mentioning the sequence of partitions
in our notation for quadratic variation.)
Since the sequence of partitions $\tau^n(\omega)$ is always nested and exhausts $\omega$
(see Lemma~\ref{lem:tau} below),
by Lemma~\ref{lem:increasing},
the limit $A(\omega)$ will be a nondecreasing function
satisfying $A_0(\omega)=0$
and $\Delta A_t(\omega)=(\Delta\omega(t))^2$ for all $t\in(0,T]$.

The following lemma lists several useful properties of $\tau^n_k$ and $D^n_k$.
\begin{lemma}\label{lem:tau}
  The functions $\tau^n_k:\Omega_{\psi}\to[0,\infty]$
  and $D^n_k:\Omega_{\psi}\to\bbbr$
  satisfy the following properties:
  \begin{enumerate}
  \item\label{it:attained}
    the infimum in the definition of $\tau^n_{k}(\omega)$ is attained,
    provided $\tau^n_k(\omega)<\infty$;
  \item\label{it:increasing}
    $\tau^n_{k}>\tau^n_{k-1}$ unless $\tau^n_{k-1}=\infty$;
  \item\label{it:valid}
    for each $n$,
    $\tau^n_k<\infty$ for only finitely many $k$;
  \item\label{it:nested}
    the sequences $\tau^n$ are nested:
    for all $n,k\in\bbbn_0$ there is $k'$ such that $\tau^{n+1}_{k'}=\tau^{n}_k$;
  \item\label{it:exhausts}
    for each $\omega\in\Omega_{\psi}$,
    the sequence of partitions $\tau^n(\omega)$ exhausts $\omega$;
  \item\label{it:bona-fide}
    $\tau^n_k$ are bona fide stopping times
    and $D^n_k$ are $\FFF_{\tau^n_k}$-measurable.
  \end{enumerate}
\end{lemma}

\ifJOURNAL
  \begin{pf}
\fi
\ifnotJOURNAL
  \begin{proof}
\fi
  Parts~\ref{it:attained} and~\ref{it:increasing}
  follow from the right-continuity of $\omega$.

  Part~\ref{it:valid} follows from the observation
  that the total number of upcrossings and downcrossings
  of the intervals of the form $(i2^{-n},(i+1)2^{-n})$, $i\in\bbbz$, by $\omega$
  over the time interval $[0,\tau^n_k]$ is at least $k-1$;
  remember that functions in $D[0,T]$ are bounded
  and can cross any non-empty interval only finitely often
  (\cite{Dellacherie/Meyer:1978}, Theorem 4.22).

  Part~\ref{it:nested} can be easily proved by induction in $k$
  if it is strengthened as follows:
  for each $n\in\bbbn_0$ there is a (unique) sequence $0=i_0\le i_1\le\cdots$
  such that $\tau^{n+1}_{i_k}=\tau^n_k$ and $\left|D^{n+1}_{i_k}-D^{n}_{k}\right|\le2^{-n-1}$.

  Part~\ref{it:exhausts} is obvious.

  It remains to prove part~\ref{it:bona-fide}.
  Fix $n\in\bbbn_0$; we will use induction in $k$
  ($\tau^n_0=0$ is obviously a stopping time).
  Fix $k\in\bbbn$ and suppose that $\tau^n_{k-1}$ is a stopping time
  and $D^n_{k-1}$ is $\FFF_{\tau^n_{k-1}}$-measurable.
  Let $t\in(0,T]$.
  The condition $\tau^n_k(\omega)\le t$ on $\omega$ can be rewritten as:
  \begin{align}
    \exists u\in(0,t)\cap\bbbq
    \;
    \exists D\in\bbbd^n:
    &\enspace\tau^n_{k-1}\le u \And D^n_{k-1}=D\And{}\\
    &\enspace\exists s\in(u,t]:
    \llbracket\omega(u),\omega(s)\rrbracket
    \cap
    (\bbbd^n\setminus\{D\})
    \ne
    \emptyset
    \label{eq:event-1}
  \end{align}
  (this uses the right-continuity of $\omega$).
  To complete the proof that $\tau^n_k$ is a stopping time,
  it suffices to prove that for any rational number $u\in(0,t)$
  and any $D\in\bbbd^n$,
  the event (\ref{eq:event-1}) belongs to $\FFF_t$.

  It is clear that (\ref{eq:event-1}) can be represented as the union
  of events of the form
  \begin{multline*} 
    \left\{
      \exists s\in(u,t]:
      \omega(u)\le D' \And \omega(s)\ge D''
    \right\}
    \\   
    =
    \left\{
      \omega(u)\le D'
    \right\}
    \cap
    \left\{
      \exists s\in(u,t]:
      \omega(s)\ge D''
    \right\}
  \end{multline*}  
  (perhaps with ``$\le$'' and ``$\ge$'' interchanged)
  for some $D',D''\in\bbbd^n$.
  Therefore, it suffices to prove that the event
  \begin{equation}\label{eq:event-2}
    \exists s\in(u,t]:
    \omega(s)\ge D''
  \end{equation}
  is in $\FFF_t$
  (the case with ``$\le$'' in place of ``$\ge$'' is treated analogously).

  The event~(\ref{eq:event-2}) is the projection onto $\Omega_{\psi}$
  of the set $A:=\{(s,\omega)\in(u,t]\times\Omega_{\psi}\st\omega(s)\ge D''\}$.
  Since c\`adl\`ag processes are progressively measurable,
  $A\in\BBB_t\times\FFF^{\circ}_t$,
  where $\BBB_t$ is the Borel $\sigma$-algebra on $[0,t]$
  and $\BBB_t\times\FFF^{\circ}_t$ is the product $\sigma$-algebra.
  This implies that the projection (\ref{eq:event-2}) is an $\FFF^{\circ}_t$-analytic set
  (according to \cite{Dellacherie/Meyer:1978}, Theorem III.13(3)),
  which in turn implies that this set is in the universal completion $\FFF_t$ of $\FFF^{\circ}_t$
  (according to \cite{Dellacherie/Meyer:1978}, Theorem III.33).


  We can see that $\tau^n_k$ is a stopping time.
  By induction in $k$, each $D^n_k$ is $\tau^n_k$-measurable:
  indeed, for $D\in\bbbd^n$,
  $D^n_k=D$ is equivalent to $\omega(\tau^n_k)\in[D,D+2^{-n})$ when $D>D^n_{k-1}$,
  is equivalent to $\omega(\tau^n_k)\in(D-2^{-n},D]$ when $D<D^n_{k-1}$,
  and is impossible when $D=D^n_{k-1}$.
\ifJOURNAL
  \qed
  \end{pf}
\fi
\ifnotJOURNAL
  \end{proof}
\fi

\subsection*{Auxiliary results}

The rest of this section is devoted to the proof of Theorem~\ref{thm:1D}.
First we notice that it suffices to prove Theorem~\ref{thm:1D}
only in the case where $\psi$ is constant, $\psi=c$ for $c>0$.
Indeed, suppose Theorem~\ref{thm:1D} holds for all $\psi=c$
and let us prove it for general nondecreasing $\psi$.
Since functions in $D[0,T]$ are bounded
and unions of countably many null sets are null,
it suffices to prove, for a fixed $L\in\bbbn$,
that typical $\omega\in\Omega_{\psi}$
satisfying $\sup_{t\in[0,T]}\lvert\omega(t)\rvert\le2^L$
have quadratic variation along $\tau(\omega)$.
The latter can be achieved by: (a) running a trading strategy
risking at most one monetary unit and bringing infinite capital
when $\omega\in\Omega_c$ for $c:=\psi(2^L)$
and $\omega$ does not have quadratic variation along $\tau(\omega)$,
and (b) stopping trading at the time
\begin{equation}\label{eq:sigma-L}
  \sigma_L(\omega)
  :=
  \inf
  \bigl\{
    t\in[0,T]\st\lvert\omega(t)\rvert\ge2^L
  \bigr\}
\end{equation}
(this is a stopping time:
see the proof of Lemma~\ref{lem:tau}\ref{it:bona-fide}).

Next we prove several auxiliary lemmas,
assuming $\psi=c$ for $c\ge1$
(there is no loss of generality in the further assumption $c\ge1$).

\begin{lemma}
  For each $n\in\bbbn$,
  the process $S^n_t:=A^n_t-A^{n-1}_t$ is a simple capital process.
\end{lemma}

\ifJOURNAL
  \begin{pf}
\fi
\ifnotJOURNAL
  \begin{proof}
\fi
  We will show that $S^n_t$ is the capital process
  of a simple trading strategy that changes its position only at times $\tau^n_k$.
  We have, for $t\in[\tau^n_k,\tau^n_{k+1}]\cap[0,T]$:
  \begin{align}
    S^n_t - S^n_{\tau^n_k}
    &=
    \left(
      A^n_t - A^{n-1}_t
    \right)
    -
    \left(
      A^n_{\tau^n_k} - A^{n-1}_{\tau^n_k}
    \right)\notag\\
    &=
    \left(
      A^n_t - A^n_{\tau^n_k}
    \right)
    -
    \left(
      A^{n-1}_t - A^{n-1}_{\tau^n_k}
    \right)\notag\\
    &=
    \Bigl(
      \omega(t) - \omega(\tau^n_k)
    \Bigr)^2
    -
    \Bigl(
      \left(
        \omega(t) - \omega(\tau)
      \right)^2
      -
      \left(
        \omega(\tau^n_k) - \omega(\tau)
      \right)^2
    \Bigr)\notag\\
    &=
    -2
    \left(
      \omega(\tau^n_k) - \omega(\tau)
    \right)
    \left(
      \omega(t) - \omega(\tau^n_k)
    \right),\label{eq:difference}
  \end{align}
  where $\tau:=\max\{\tau^{n-1}_{k'}\st\tau^{n-1}_{k'}\le\tau^n_k\}$.
  Therefore,
  it suffices to take position
  $
    -2
    \left(
      \omega(\tau^n_k) - \omega(\tau)
    \right)
  $
  at time $\tau^n_k$.
\ifJOURNAL
  \qed
  \end{pf}
\fi
\ifnotJOURNAL
  \end{proof}
\fi

The proof of Theorem~\ref{thm:1D} involves the following simple capital process
(based on a standard idea going back
to at least Kolmogorov \cite{Kolmogorov:1929LLN}):
\begin{equation}\label{eq:Kolmogorov}
  U^n_t
  :=
  2^{-2n+8} c^2
  +
  n^4 2^{-2n}
  +
  (S^n_t)^2
  -
  \sum_{k=1}^{\infty}
  \left(
    S^n_{\tau^n_k\wedge t}
    -
    S^n_{\tau^n_{k-1}\wedge t}
  \right)^2.
\end{equation}

\begin{lemma} 
  For each $n\in\bbbn$,
  $U^n_t$ is indeed a simple capital process.
\end{lemma}

\ifJOURNAL
  \begin{pf}
\fi
\ifnotJOURNAL
  \begin{proof}
\fi
  As in the previous lemma,
  the position is only changed at times $\tau^n_k$.
  When $t\in[\tau^n_k,\tau^n_{k+1}]\cap[0,T]$, we have:
  \begin{align*}
    U^n_t - U^n_{\tau^n_k}
    &=
    \left(
      \left(S^n_t\right)^2
      -
      \sum_{i=1}^{k}
      \left(
        S^n_{\tau^n_i}
        -
        S^n_{\tau^n_{i-1}}
      \right)^2
      -
      \left(
        S^n_{t}
        -
        S^n_{\tau^n_k}
      \right)^2
    \right)
    \\       
    &\quad   
    -
    \left(
      \left(S^n_{\tau^n_k}\right)^2
      -
      \sum_{i=1}^{k}
      \left(
        S^n_{\tau^n_i}
        -
        S^n_{\tau^n_{i-1}}
      \right)^2
    \right)
    \\
    &=
    2 S^n_{\tau^n_k}
    \left(
      S^n_{t}
      -
      S^n_{\tau^n_k}
    \right)
    \\&   
    =
    -4 S^n_{\tau^n_k}
    \left(
      \omega(\tau^n_k)
      -
      \omega(\tau)
    \right)
    \left(
      \omega(t)
      -
      \omega(\tau^n_k)
    \right)
  \end{align*}
  (the last equality follows from (\ref{eq:difference})).
  Therefore, it suffices to take position
  $
    -4 S^n_{\tau^n_k}
    \left(
      \omega(\tau^n_k)
      -
      \omega(\tau)
    \right)
  $
  at time $\tau^n_k$.
\ifJOURNAL
  \qed
  \end{pf}
\fi
\ifnotJOURNAL
  \end{proof}
\fi

Set
\begin{equation}\label{eq:sigma-n}
  \sigma^n
  :=
  \min
  \left\{
    \tau^n_K
    \biggm|
    \sum_{k=1}^{K}
    \left(
      S^n_{\tau^n_k}
      -
      S^n_{\tau^n_{k-1}}
    \right)^2
    >
    n^4 2^{-2n}
  \right\}.
\end{equation}
We will be interested in the nonnegative simple capital process
$U^n_{\sigma^n\wedge t}$;
its nonnegativity (on $\Omega_{\psi}=\Omega_c$) follows from
$(S^n_t - S^n_{\tau^n_k})^2\le2^{-2n+8}c^2$,
where $t\in[\tau^n_k,\tau^n_{k+1}]\cap[0,T]$,
which in turn follows from (\ref{eq:difference}):
\begin{multline*} 
  \left|
    S^n_t - S^n_{\tau^n_k}
  \right|
  =
  2
  \left|
    \omega(\tau^n_k) - \omega(\tau)
  \right|
  \left|
    \omega(t) - \omega(\tau^n_k)
  \right|
  \\  
  \le
  2
  \left(
    2\times2^{-n+1}
  \right)
  \left(
    c + 2\times2^{-n}
  \right)
  \le
  2^{-n+4}c.
\end{multline*} 

To analyse the process $U^n_{\sigma^n\wedge t}$
we will need a probability-free version
of Doob's upcrossing inequality (Lemma~\ref{lem:Doob} below)
and its corollary (Lemma~\ref{lem:Stricker})
obtained by a method proposed by Bruneau \cite{Bruneau:1979}
and developed and simplified in \cite{Stricker:1979} and \cite{\CTV}.

Let $\MM_t^{(a,b)}(f)$ (resp.\ $\DD_t^{(a,b)}(f)$)
be the number of upcrossings (resp.\ downcrossings) of an open interval $(a,b)$
by a function $f:[0,T]\to\bbbr$ during the time interval $[0,t]$.
For each $h>0$ set
\begin{equation*}
  \MM_t(f,h)
  :=
  \sum_{k\in\bbbz}
  \MM_t^{(k h,(k+1)h)}(f),
  \quad
  \DD_t(f,h)
  :=
  \sum_{k\in\bbbz}
  \DD_t^{(k h,(k+1)h)}(f).
\end{equation*}
Remember that $\sigma_L$ is the stopping time defined by (\ref{eq:sigma-L}).

\begin{lemma}\label{lem:Doob}
  Let $L\in\bbbn$ and $(a,b)\subseteq(-2^L,2^L)$ be a non-empty interval.
  There exists a nonnegative simple capital process $V$
  that starts from $V_0=a+2^L+c$ and satisfies
  \begin{equation}\label{eq:Doob}
    V_t(\omega)
    \ge
    (b-a)\MM_t^{(a,b)}(\omega)
  \end{equation}
  for all $\omega\in\Omega_{\psi}=\Omega_c$ and all $t\in[0,\sigma_L(\omega)]$.
\end{lemma}

\noindent
(Remember that $V$ being nonnegative means that $V_t(\omega)\ge0$
for all $t$ and $\omega\in\Omega_{\psi}$.)

\ifJOURNAL
  \begin{pf}
\fi
\ifnotJOURNAL
  \begin{proof}
\fi
  The following standard argument will be easy to formalize.
  A simple trading strategy $G$ leading to $V$
  can be defined as follows.
  The initial capital is $a+2^L+c$.
  At first $G$ takes position $0$.
  When $\omega$ first hits $(-\infty,a]$,
  $G$ takes position $1$ until $\omega$ hits $[b,\infty)$,
  at which point $G$ takes position $0$;
  after $\omega$ hits $(-\infty,a]$,
  $G$ maintains position $1$ until $\omega$ hits $[b,\infty)$,
  at which point $G$ takes position $0$; etc.
  The only exception is that trading is stopped at time $\sigma_L$:
  the position at that time becomes 0 and stays 0 afterwards.
  (The essential bit is that trading should be stopped
  when $\omega$ hits $(-\infty,-2^L]$, if this ever happens.)
  Since $\omega$'s jumps never exceed $c$ in absolute value,
  the process $S$ will be nonnegative.

  Formally, we define $\tau_1:=\inf\{t\in[0,T]\st\omega(t)\in(-\infty,a]\}$
  and, for $n=2,3,\ldots$,
  \begin{equation*}
    \tau_n
    :=
    \inf\{t\in[0,T]\st t>\tau_{n-1} \And \omega(t)\in I_n\},
  \end{equation*}
  where $I_n:=[b,\infty)$ for even $n$
  and $I_n:=(-\infty,a]$ for odd $n$.
  Each $\tau_n$ is a stopping time:
  this can be shown analogously
  to the proof of Lemma~\ref{lem:tau}\ref{it:bona-fide}.
  We have $\tau_n(\omega)<\infty$ for only finitely many $n$
  since $\omega$ is c\`adl\`ag:
  see, e.g., \cite{Dellacherie/Meyer:1978}, Theorem IV.22.

  Since $\omega$ is a right-continuous function
  and $(-\infty,a]$ and $[b,\infty)$ are closed sets,
  the infimum in the definition of $\tau_n$, $n=1,2,\ldots$, is attained
  when $\tau_n<\infty$.
  Therefore,
  $\omega(\tau_1)\le a$,
  $\omega(\tau_2)\ge b$,
  $\omega(\tau_3)\le a$,
  $\omega(\tau_4)\ge b$, and so on while $\tau_n<\infty$.
  Set, for $n=1,2,\ldots$,
  $$
    \tau'_n
    :=
    \begin{cases}
      \tau_n & \text{if $\tau_n<\sigma_L$}\\
      \sigma_L & \text{if $\tau_{n-1}<\sigma_L\le\tau_n$}\\
      \infty & \text{otherwise}
    \end{cases}
  $$
  (the inequality $\tau_{n-1}<\sigma_L$ is considered to be true
  when $n=1$).
  The position taken by $G$ at the time $\tau'_n$,
  $n=1,2,\ldots$,
  is
  $$
    h_n
    :=
    \begin{cases}
      1 & \text{if $\tau'_n<\sigma_L$ and $n$ is odd}\\
      0 & \text{otherwise},
    \end{cases}
  $$
  and the initial capital is $a+2^L+c$.
  Let $t\in[0,\sigma_L]$ and $n$ be the largest integer such that $\tau_n\le t$
  (with $n:=0$ when $\tau_1>t$; the formulas below will never involve $\tau_0$).
  Now we obtain from (\ref{eq:simple-capital}):
  if $n$ is even,
  \begin{align*}
    V_t(\omega)
    &=
    V_0
    +
    (\omega(\tau_2)-\omega(\tau_1))
    +
    (\omega(\tau_4)-\omega(\tau_3))
    +\cdots+
    (\omega(\tau_n)-\omega(\tau_{n-1}))\\
    &\ge
    a + 2^L + c + (b-a)\MM_t^{(a,b)}(\omega),
  \end{align*}
  and if $n$ is odd,
  \begin{align*}
    V_t(\omega)
    &=
    V_0
    +
    (\omega(\tau_2)-\omega(\tau_1))
    +
    (\omega(\tau_4)-\omega(\tau_3))
    +\cdots+
    (\omega(\tau_{n-1})-\omega(\tau_{n-2}))
    \\&\quad   
    +
    (\omega(t)-\omega(\tau_{n}))\\
    &\ge
    a + 2^L + c + (b-a)\MM_t^{(a,b)}(\omega) 
    +
    (\omega(t)-\omega(\tau_{n}))\\
    &\ge
    a + 2^L + c + (b-a)\MM_t^{(a,b)}(\omega) 
    +
    (-2^L - c - a)
    =
    (b-a)\MM_t^{(a,b)}(\omega);
  \end{align*}
  in both cases, (\ref{eq:Doob}) holds.
  In particular, $V_t(\omega)$ is nonnegative.
\ifJOURNAL
  \qed
  \end{pf}
\fi
\ifnotJOURNAL
  \end{proof}
\fi

\begin{lemma}\label{lem:Stricker}
  Let $L\in\bbbn$.
  For each $n\in\bbbn$,
  there exists a nonnegative simple capital process $V$ such that
  $V_0\le2^L+c$
  and $V_t(\omega)\ge 2^{-2n-L-1}\MM_{t-}(\omega,2^{-n})$
  for all $\omega\in\Omega_c$ and all $t\in[0,\sigma_L(\omega)]$.
\end{lemma}

\ifJOURNAL
  \begin{pf}
\fi
\ifnotJOURNAL
  \begin{proof}
\fi
  By Lemma~\ref{lem:Doob},
  for each $k\in\{-2^{L+n},\ldots,2^{L+n}-1\}$
  there exists a nonnegative simple capital process $V^k$
  that starts from $k2^{-n}+2^L+c$ and satisfies
  \begin{equation*}
    V^{k}_t(\omega)
    \ge
    2^{-n}
    \MM_t^{(k2^{-n},(k+1)2^{-n})} (\omega)
  \end{equation*}
  for all $t\in[0,\sigma_L(\omega)]$.
  Summing $2^{-L-n-1}V^{k}$ over $k=-2^{L+n},\ldots,2^{L+n}-1$
  (i.e., averaging $V^k$),
  we obtain a nonnegative simple capital process $V$ such that
  \begin{align*}
    V_0
    &=
    2^{-L-n-1}
    \sum_{k=-2^{L+n}}^{2^{L+n}-1}
    k2^{-n}
    + 2^L + c
    \le
    2^L + c,\\
    V_{t}(\omega)
    &\ge
    2^{-2n-L-1}
    \MM_t(\omega,2^{-n})
    \text{ for all $t<\sigma_L$};
  \end{align*}
  at $t=\sigma_L(\omega)$ we are only guaranteed to have
  $
    V_{t}(\omega)
    \ge
    2^{-2n-L-1}
    \MM_{t-}(\omega,2^{-n})
  $
  because of the possibility $\omega(t)\notin[-2^L,2^L]$.
\ifJOURNAL
  \qed
  \end{pf}
\fi
\ifnotJOURNAL
  \end{proof}
\fi

\begin{corollary}\label{cor:Stricker}
  Let $L\in\bbbn$.
  For typical $\omega\in\Omega_c$,
  from some $n\in\bbbn$ on,
  $\MM_{\sigma_L-}(\omega,2^{-n})\le n^2 2^{2n}$
  and $\DD_{\sigma_L-}(\omega,2^{-n})\le n^2 2^{2n}$.
\end{corollary}
\ifJOURNAL
  \begin{pf}
\fi
\ifnotJOURNAL
  \begin{proof}
\fi
  It suffices to prove that
  $\MM_{\sigma_L-}(\omega,2^{-n})\le n^{1.5} 2^{2n}$ from some $n$ on
  (since $\DD_{\sigma_L-}(\omega,2^{-n})\le\MM_{\sigma_L-}(\omega,2^{-n})+2^{n+L}$).
  Consider the event $E$ that the inequality
  $\MM_{\sigma_L-}(\omega,2^{-n}) > n^{1.5} 2^{2n}$
  holds for infinitely many $n$.
  (This is the complementary event to the event
  that we are proving to be almost certain.)
  For each $n\in\bbbn$ let $V^n_t$ be a nonnegative simple capital process
  such that $V^n_0\le2^L+c$ and
  $V^n_{\sigma_L} \ge 2^{-2n-L-1}\MM_{\sigma_L-}(\omega,2^{-n})$
  (see Lemma~\ref{lem:Stricker}).
  Then, for infinitely many $n$, $V^n_{\sigma_L} \ge 2^{-L-1}n^{1.5}$ on $E$.
  To see that the event $E$ is null,
  it suffices to notice that the process $\sum_n n^{-1.5} V^n$
  starts from $V_0<\infty$ and satisfies $V_{\sigma_L}=\infty$ on~$E$.
\ifJOURNAL
  \qed
  \end{pf}
\fi
\ifnotJOURNAL
  \end{proof}
\fi

\subsection*{Proof of Theorem \ref{thm:1D} (for $\psi=c$)}

Let $L\in\bbbn$.
For all $\omega\in\Omega_c$ and for $t:=\sigma_L(\omega)$,
the infinite sum in (\ref{eq:Kolmogorov}) can be bounded above as follows
(cf.\ (\ref{eq:difference})):
\begin{align}
  &\hspace*{-1cm}   
  \sum_{k=0}^{\infty}
  \left(
    S^n_{\tau^n_{k+1}\wedge\sigma_L}
    -
    S^n_{\tau^n_{k}\wedge\sigma_L}
  \right)^2
  \notag\\   
  &=
  2^2
  \sum_{k=0}^{\infty}
  \left(
    \omega(\tau^n_k\wedge\sigma_L) - \omega(\tau\wedge\sigma_L)
  \right)^2
  \left(
    \omega(\tau^n_{k+1}\wedge\sigma_L) - \omega(\tau^n_k\wedge\sigma_L)
  \right)^2
  \notag\\
  &\le
  2^{6-2n}
  \sum_{k=0}^{\infty}
  \left(
    \omega(\tau^n_{k+1}\wedge\sigma_L) - \omega(\tau^n_k\wedge\sigma_L)
  \right)^2.
  \label{eq:sum}
\end{align}

Fix an $\omega\in\Omega_c$ such that, from some $n\in\bbbn$ on,
$\MM_{\sigma_L-}(\omega,2^{-n})\le n^2 2^{2n}$
and $\DD_{\sigma_L-}(\omega,2^{-n})\le n^2 2^{2n}$;
by Corollary~\ref{cor:Stricker},
this condition is satisfied for typical $\omega\in\Omega_c$.
Let $N\in\{2,3,\ldots\}$ be so large that
$\MM_{\sigma_L-}(\omega,2^{-n})\le n^2 2^{2n}$
and $\DD_{\sigma_L-}(\omega,2^{-n})\le n^2 2^{2n}$
for all $n\ge N$.
In particular, for all $m\ge N$:
\begin{itemize}
\item
  the number of $k\in\bbbn_0$ such that $\tau^m_k<\infty$
  does not exceed $2m^22^{2m}+2\le3m^22^{2m}$;
\item
  the number of jumps of $\omega$ of size $2^{-m+1}$ or more
  over the time interval $[0,\sigma_L(\omega))$
  does not exceed $2m^22^{2m}$.
  (By the \emph{size} of a jump of $\omega$ at $t\in(0,T]$
  we mean its absolute value $\lvert\Delta\omega(t)\rvert$.)
\end{itemize}
We will write $t\in\tau^n$ to mean that $t$ is an element of the partition $\tau^n$:
there exists $k\in\bbbn_0$ such that $t=\tau^n_k$.

Suppose $n\ge N$.
We will bound different addends in (\ref{eq:sum}) in different ways:
\renewcommand{\theenumi}{(\Alph{enumi})}
\begin{enumerate}
\item\label{it:first}
  If $\tau^n_k\le\sigma_L\le\tau^n_{k+1}$, we use the trivial bound
  $$
    \left|
      \omega(\tau^n_{k+1}\wedge\sigma_L) - \omega(\tau^n_k\wedge\sigma_L)
    \right|
    \le
    2^{1-n}+c.
  $$
  There are at most two such addends.
  We ignore the zero addends for which $\tau^n_k>\sigma_L$,
  and so assume $\tau^n_{k+1}<\sigma_L$ in the rest of this list.
\item\label{it:not-in}
  If $\tau^n_{k+1}\notin\tau^{n-1}$,
  $$
    \left|
      \omega(\tau^n_{k+1}\wedge\sigma_L) - \omega(\tau^n_k\wedge\sigma_L)
    \right|
    \le
    2^{2-n}.
  $$
  The number of such $k$ is at most $3n^22^{2n}$.
\item\label{it:small-jump}
  If $\tau^n_{k+1}\in\tau^{n-1}$
  and the size of the jump of $\omega$ at $\tau^n_{k+1}$ is below $2^{-n+1}$,
  $$
    \left|
      \omega(\tau^n_{k+1}\wedge\sigma_L) - \omega(\tau^n_k\wedge\sigma_L)
    \right|
    \le
    2^{1-n} + 2^{-n+1}
    =
    2^{2-n}.
  $$
  The number of such $k$
  is at most $3n^22^{2n}$.
\item\label{it:intermediate-jump}
  If $\tau^n_{k+1}\in\tau^{n-1}$
  and the size of the jump of $\omega$ at $\tau^n_{k+1}$
  belongs to $[2^{-m+1},2^{-m+2})$,
  where $m\in\{n,n-1,\ldots,N\}$,
  $$
    \left|
      \omega(\tau^n_{k+1}\wedge\sigma_L) - \omega(\tau^n_k\wedge\sigma_L)
    \right|
    \le
    2^{1-n} + 2^{-m+2}.
  $$
  The number of such $k$ is at most $2m^22^{2m}$.
\item\label{it:large-jump}
  If $\tau^n_{k+1}\in\tau^{n-1}$
  and the size of the jump of $\omega$ at $\tau^n_{k+1}$ is $2^{-N+2}$ or more,
  we will use the trivial bound
  $$
    \left|
      \omega(\tau^n_{k+1}\wedge\sigma_L) - \omega(\tau^n_k\wedge\sigma_L)
    \right|
    \le
    2^{1-n}+c.
  $$
  The number of such $k$ is bounded by a constant $C$
  (it is a constant in the sense of not depending on $n$,
  but it depends on $\omega$ and $L$).
\end{enumerate}
\renewcommand{\theenumi}{(\alph{enumi})}
Now we can bound the last sum in (\ref{eq:sum}) as follows:
\begin{multline*}
  \sum_{k=0}^{\infty}
  \left(
    \omega(\tau^n_{k+1}\wedge\sigma_L) - \omega(\tau^n_k\wedge\sigma_L)
  \right)^2
  \le
  2 (2^{1-n}+c)^2
  +
  3n^2 2^{2n}
  \left(
    2^{2-n}
  \right)^2\\
  +
  3n^2 2^{2n}
  \left(
    2^{2-n}
  \right)^2
  +
  \sum_{m=N}^n
  2m^2
  2^{2m}
  \left(
    2^{1-n}
    +
    2^{-m+2}
  \right)^2
  +
  C (2^{1-n}+c)^2
  \le
  2^{-6}n^4,
\end{multline*}
the last inequality being true from some $n$ on.
Remembering the definitions (\ref{eq:Kolmogorov}) and (\ref{eq:sigma-n})
of $U^n_t$ and $\sigma^n$,
we can see that, from some $n$ on,
$\sigma^n(\omega)>\sigma_L(\omega)$ and, therefore,
\begin{equation}\label{eq:U-1}
  U^n_{\sigma^n\wedge t}(\omega)
  =
  U^n_{t}(\omega)
  \ge
  (S^n_t(\omega))^2
\end{equation}
for all $t\le\sigma_L(\omega)$.

We have shown that, for typical $\omega\in\Omega_c$, from some $n$ on,
(\ref{eq:U-1}) holds for all $t\le\sigma_L(\omega)$.
On the other hand,
for typical $\omega\in\Omega_c$ from some $n$ on
we will have
\begin{equation}\label{eq:U-2}
  \forall t\in[0,\sigma_L(\omega)]:
  U^n_{\sigma^n\wedge t}(\omega) < n^62^{-2n}.
\end{equation}
Indeed, if $V^n_t$ is the process $U^n_{\sigma_L\wedge\sigma^n\wedge t}$
stopped when it reaches the value $n^62^{-2n}$,
the nonnegative capital process
$$
  V_t
  :=
  \sum_{n=1}^{\infty}
  n^{-6}
  2^{2n}
  V^n_{t}
$$
will have a finite initial value and satisfy $V_T(\omega)=\infty$
for $\omega$ such that,
for infinitely many $n$,
$\exists t\in[0,\sigma_L(\omega)]:U^n_{\sigma^n\wedge t}(\omega)\ge n^62^{-2n}$.
Combining (\ref{eq:U-1}) and (\ref{eq:U-2}),
we obtain that, for typical $\omega\in\Omega_c$,
from some $n$ on,
we have $\lvert S^n_t\rvert < n^32^{-n}$
for all $t\in[0,\sigma_L(\omega)]$.

We can see that,
for typical $\omega\in\Omega_c$,
the uniform distance between $A^{n-1}$ and $A^n$
does not exceed $n^32^{-n}$ from some $n$ on
provided $\sigma_L=T$.
Therefore, for typical $\omega\in\Omega_c$
the sequence $A^n$ is convergent in the uniform metric
provided $\sigma_L=T$.
Since the union of countably many null sets is a null set,
we can omit ``provided $\sigma_L=T$''.

\section{Multidimensional case}
\label{sec:multidim}

The goal of this section is to establish the existence of quadratic covariation
between different price paths.
We will be using a standard expression of quadratic covariation
in terms of quadratic variation
(the ``polarization identity'',
used in \cite{Follmer:1981}, Remark~1 on p.~147).

Let $\omega^m\in\Omega_{\psi}$, $m=1,\ldots,M$.
In our informal discussions
we will assume that $\omega^m$, $m=1,\ldots,M$,
are the price paths of all securities traded in a financial market.
We will write $\omega$ for the function whose value $\omega(t)$
at time $t\in[0,T]$ is the vector
$(\omega^1(t),\ldots,\omega^M(t))\in\bbbr^M$;
the set of all such functions $\omega$ with components $\omega^m\in\Omega_{\psi}$
will be denoted $\Omega_{\psi}^M$.

In this section we set $\Omega:=\Omega_{\psi}^M$
(little, however, will depend on this specific definition of $\Omega$,
and the definitions and statements below work for a wide class of $\Omega$).
The $\sigma$-algebra $\FFF^{\circ}_t$, $t\in[0,T]$,
is the smallest $\sigma$-algebra on $\Omega$
that makes all functions $\omega\in\Omega\mapsto\omega^m(s)$,
where $m\in\{1,\ldots,M\}$ and $s\in[0,t]$, measurable.
The $\sigma$-algebras $\FFF_t$, processes, stopping times $\tau$,
and $\sigma$-algebras $\FFF_{\tau}$ on $\Omega$
are defined for the sample space $\Omega$
in the same way as for the sample space $\Omega_{\psi}$:
replacing $\Omega_{\psi}$ by $\Omega$ is the only change.
The other definitions of Section~\ref{sec:definitions}
carry over to the case of $\Omega=\Omega_{\psi}^M$,
with the following changes.
In the definition of simple trading strategies,
the bounded $\FFF_{\tau_n}$-measurable functions $h_n$
now take values in $\bbbr^M$,
and the definition (\ref{eq:simple-capital})
of a simple capital process now becomes
\begin{equation*}
  \K^{G,\alpha}_t(\omega)
  :=
  \alpha
  +
  \sum_{n=1}^{\infty}
  h_n(\omega)
  \cdot
  \bigl(
    \omega(\tau_{n+1}\wedge t)-\omega(\tau_n\wedge t)
  \bigr),
  \quad
  t\in[0,T],
\end{equation*}
where ``$\cdot$'' stands for dot product in $\bbbr^M$.

For all $\omega\in\Omega$ and $n\in\bbbn_0$,
let us define a sequence $\tau^n_k(\omega)$, $k=0,1,2,\ldots$,
as the finite set
\begin{multline*} 
  \bigl\{
    \tau^n_k(\omega^m)\st k\in\bbbn_0,\;m\in\{1,\ldots,M\}
  \bigr\}
  \\  
  \cup
  \bigl\{
    \tau^n_k(\omega^m+\omega^l)\st k\in\bbbn_0,\;
    m\in\{1,\ldots,M\},\;l\in\{1,\ldots,M\}
  \bigr\}
\end{multline*}  
ordered in the increasing order
(with repetitions removed, by the definition of a set);
to make the sequence $\tau^n(\omega)$ infinite,
we complement it by $\infty,\infty,\ldots$ on the right.
It is clear that $\tau^n_k(\omega)$ as functions of $\omega\in\Omega$
are stopping times.

The sequence $(\tau^n_0(\omega),\tau^n_1(\omega),\ldots)$
will be denoted $\tau^n(\omega)$;
for a fixed $\omega$, this is a partition of $[0,T]$.
By $\tau(\omega)$ we denote the sequence of partitions $\tau^n(\omega)$, $n\in\bbbn_0$.

\begin{theorem}\label{thm:multidim}
  For typical $\omega\in\Omega=\Omega^M_{\psi}$,
  each $\omega^m$, $m\in\{1,\ldots,M\}$, has quadratic variation along $\tau(\omega)$,
  and each $\omega^m+\omega^l$, $(m,l)\in\{1,\ldots,M\}^2$,
  has quadratic variation along $\tau(\omega)$.
\end{theorem}

Before proving Theorem~\ref{thm:multidim},
let us see how to use it to define quadratic covariation.
We will use the notation $[\omega^m]$
for the quadratic variation of $\omega^m$,
$m\in\{1,\ldots,M\}$,
and $[\omega^m+\omega^l]$ for the quadratic variation of $\omega^m+\omega^l$,
$(m,l)\in\{1,\ldots,M\}^2$,
along $\tau(\omega)$.
Now we define the quadratic covariation processes
between different price paths by
\begin{equation}   
  \label{eq:covariation}
  [\omega]^{m,l}_t
  :=
  \frac12
  \left(
    [\omega^m+\omega^l]_t
    -
    [\omega^m]_t
    -
    [\omega^l]_t
  \right),
  \qquad  
  t\in[0,T],\;
  (m,l)\in\{1,\ldots,M\}^2;
\end{equation}    
they exist for typical $\omega\in\Omega$.
Notice that $[\omega]^{m,m}=[\omega^m]$ and that
$$
  [\omega]^{m,l}_t
  =
  \left(
    [\omega]^{m,l}
  \right)\cont_t
  +
  \sum_{s\in(0,t]}
  \Delta\omega^m(s)
  \Delta\omega^l(s),
$$
where $f\cont$ stands for the continuous part of a function $f$.

\subsection*{Proof of Theorem~\ref{thm:multidim}}

We have already proved that $\omega^m$ have quadratic variation,
but along a different sequence of partitions,
$\tau(\omega^m)$ rather than $\tau(\omega)$.
Now we stop and change our positions
as soon as any security in the market or the sum of two securities
significantly change their value,
so we have one sequence of partitions for the whole market.
But the argument of Section~\ref{sec:QV-proof} still works,
as it depends on relatively few properties of the sequence of partitions,
which are still satisfied.

To see what properties of the array $\tau^n_k$ of stopping times were used
in Section~\ref{sec:QV-proof},
let $\iota^n_k$ ($n\in\bbbn_0$, $k\in\bbbn_0$)
be a general array of stopping times on $\Omega$ of the same kind as $\tau^n_k$:
$\iota^n(\omega)$, $n\in\bbbn_0$, are nested partitions of $[0,T]$ for each $\omega\in\Omega$
and all $\iota^n_k$ are stopping times;
in particular, $0=\iota^n_0\le\iota^n_1\le\cdots$
and all $\iota^n_k$ take values in $[0,T]\cup\{\infty\}$.
We say that the sequence $\iota$ of random (i.e., depending on $\omega\in\Omega$)
nested partitions $\iota^n$ is \emph{of dyadic type} for a process $W$
(such as $\omega^m$ or $\omega^m+\omega^l$) if
there exist a polynomial $p$ and a constant $C>0$ such that:
\begin{enumerate}
\item\label{it:below}
  For all $\omega\in\Omega$, all $n\in\bbbn_0$, and all $0\le s<t\le T$
  such that $\lvert W_t-W_s\rvert>C2^{-n}$,
  there exists $k$ such that $\iota^n_k(\omega)\in(s,t]$.
\item\label{it:above}
  For typical $\omega\in\Omega$, from some $n$ on,
  the number of $k$ such that $\iota^n_k(\omega)<\infty$
  is at most $p(n)2^{2n}$.
\end{enumerate}
Intuitively, \ref{it:below} says that the partitions are sufficiently fine,
and \ref{it:above} says that they are not too fine.
Condition~\ref{it:below} implies that $\iota(\omega)$ should exhaust $\omega$
and can be regarded as a quantitative version of this requirement.

A process $W$ \emph{has moderate jumps} if
there exists a nondecreasing function $\phi:[0,\infty)\to(0,\infty)$ such that,
for each $\omega\in\Omega$,
\begin{equation*}
  \forall t\in(0,T]:
  \left|
    \Delta W_t(\omega)
  \right|
  \le
  \phi
  \left(
    \sup_{s\in[0,t)}
    \lvert W_s(\omega)\rvert
  \right)
\end{equation*}
(cf.\ (\ref{eq:moderate})).
We let $W(\omega)$, where $W$ is a process and $\omega\in\Omega$,
stand for the function $t\in[0,T]\mapsto W_t(\omega)\in\bbbr$.
The quadratic variation of $W(\omega)$ along $\iota(\omega)$ is defined as before
(cf.\ (\ref{eq:A1}) and (\ref{eq:A2})) with $W(\omega)$ in place of $\omega$,
namely, as the uniform limit as $n\to\infty$ of
\begin{equation*}
  A^{n,\tau}_t(\omega)
  :=
  \sum_{k=1}^{\infty}
  \left(
    W_{\iota^n_k\wedge t}(\omega)
    -
    W_{\iota^n_{k-1}\wedge t}(\omega)
  \right)^2.
\end{equation*}

We saw in Section~\ref{sec:QV-proof} that the sequence $\tau(\omega)$
of partitions defined earlier in this section is of dyadic type for all $\omega\in\Omega$
(with $C:=2$ and $p(n):=(M^2+M)3n^2$ sufficient).
This shows, in combination with Proposition~\ref{prop:invariance} below,
that the following proposition generalizes Theorem~\ref{thm:1D}.

\begin{proposition}\label{prop:1D}
  If an array $\iota^n_k$ of stopping times on $\Omega$ is of dyadic type
  for a simple capital process $W$ that has moderate jumps,
  $W(\omega)$ has quadratic variation along $\iota(\omega)$
  for typical $\omega\in\Omega$.
\end{proposition}

\ifJOURNAL
  \begin{pf}
\fi
\ifnotJOURNAL
  \begin{proof}
\fi
  It is easy to check that the proof of Theorem~\ref{thm:1D},
  with suitable changes,
  still works in our current more general situation;
  now we consider paths $W(\omega)$ instead of $\omega$.

  Condition~\ref{it:below} (with $C=2$) 
  was used in establishing the nonnegativity of the simple capital process (\ref{eq:Kolmogorov})
  over $[0,\sigma^n\wedge T]$
  and our upper bound on the infinite sum (\ref{eq:sum})
  in (\ref{eq:Kolmogorov}).
  For our arguments to go through,
  the definition of the simple capital process (\ref{eq:Kolmogorov})
  should be modified:
  the addend $2^{-2n+8}c^2$ should be replaced by a constant
  (depending only on $c$ and $C$) times $2^{-2n}$,
  and the polynomial $n^4$ in the addend $n^42^{-2n}$
  should be replaced by a different polynomial (depending on $p$).
  To upper bound (\ref{eq:sum}),
  we have used (in the cases \ref{it:first}--\ref{it:large-jump} of Section~\ref{sec:QV-proof})
  the inequality
  $$
    \left|
      \omega(\tau^n_{k+1}(\omega)) - \omega(\tau^n_k(\omega))
    \right|
    \le
    2^{1-n}
    +
    \left|
      \Delta\omega(\tau^n_{k+1})
    \right|
  $$
  assuming $\tau^n_{k+1}(\omega)<\sigma_L(\omega)$.
  This inequality continues to hold with $2^{1-n}$ replaced by $C2^{-n}$
  in our current situation:
  indeed, condition~\ref{it:below} 
  implies that 
  $$
    \left|
      \omega(\iota^n_{k+1}(\omega)-) - \omega(\iota^n_k(\omega))
    \right|
    \le
    C 2^{-n}.
  $$

  Condition~\ref{it:above}
  can be used in the cases \ref{it:not-in}--\ref{it:intermediate-jump} in Section~\ref{sec:QV-proof}
  to bound the number of $k$ covered by each of those cases.
  We need the weaker requirement that $\iota^n_k<\infty$ for only finitely many $k$
  also to satisfy a requirement in the definition of simple trading strategies.
\ifJOURNAL
  \qed
  \end{pf}
\fi
\ifnotJOURNAL
  \end{proof}
\fi

Let $m,l\in\{1,\ldots,M\}$.
The processes $\omega^m$ and $\omega^m+\omega^l$ have moderate jumps
(with $\phi:=\psi$ and $\phi:=2\psi$, respectively).
Both $\omega^m$ and $\omega^m+\omega^l$ are simple capital processes.
We can see that, for typical $\omega\in\Omega$,
$\omega^m$ and $\omega^m+\omega^l$ have quadratic variation along $\tau(\omega)$.

The proof of Theorem~\ref{thm:multidim} is now complete.
However, there remains the question of invariance of our definitions;
e.g., is the quadratic variation $A^{\tau(\omega)}(\omega^m)$ of $\omega^m$
along $\tau(\omega)$ (our new definition)
the same function as the quadratic variation $A^{\tau(\omega^m)}(\omega^m)$ of $\omega^m$
along $\tau(\omega^m)$ (our old definition)?
The proof of Theorem~\ref{thm:1D} shows that they are.
For simplicity, we will only spell out the argument
in the framework of Proposition~\ref{prop:1D}.

\begin{proposition}\label{prop:invariance}
  If arrays $\iota^n_k$ and $\kappa^n_k$
  of stopping times on $\Omega$ are of dyadic type for a simple capital process $W$,
  $W(\omega)$ has the same quadratic variation
  along $\iota(\omega)$ and along $\kappa(\omega)$
  for typical $\omega\in\Omega$.
\end{proposition}

\ifJOURNAL
  \begin{pf}
\fi
\ifnotJOURNAL
  \begin{proof}
\fi
  Without loss of generality we can assume that $\iota$ is nested in $\kappa$:
  for all $\omega\in\Omega$, $n\in\bbbn_0$, and $k\in\bbbn_0$,
  there exists $k'\in\bbbn_0$ such that
  $\kappa^n_{k'}(\omega)=\iota^n_{k}(\omega)$.
  (Indeed, this special case will imply that general $\iota$ and $\kappa$
  lead to the same quadratic variation as their union.)
  The argument in the proof of Theorem~\ref{thm:1D}
  (applied to $\iota(\omega)$ and $\kappa(\omega)$
  in place of $\tau^{n-1}(\omega)$ and $\tau^{n}(\omega)$,
  so that $S^n_t(\omega):=A^{n,\kappa}_t(\omega)-A^{n,\iota}_t(\omega)$)
  shows that, for typical $\omega\in\Omega$, the uniform distance
  between $A^{n,\iota}(\omega)$ and $A^{n,\kappa}(\omega)$
  converges to 0 (exponentially fast) as $n\to\infty$.
  Therefore, the uniform limits of
  $A^{n,\iota}(\omega)$ and $A^{n,\kappa}(\omega)$
  coincide.
\ifJOURNAL
  \qed
  \end{pf}
\fi
\ifnotJOURNAL
  \end{proof}
\fi

\section{F\"ollmer's and Norvai{\u s}a's quadratic variation}
\label{sec:other-definitions}

In this section we adapt F\"ollmer's \cite{Follmer:1981}
and Norvai{\u s}a's \cite{Norvaisa:2001} definitions of quadratic variation
to our framework (in particular, to our bounded time interval $[0,T]$).
Let $\omega\in D[0,T]$ and $\pi$ be a nested sequence of partitions that exhausts $\omega$.

We say that $\omega$ \emph{has F\"ollmer's quadratic variation along $\pi$}
if
the sequence of finite measures
$$
  \xi^n
  :=
  \sum_{k=1}^{\infty}
  \left(
    \omega(\pi_{k}^n\wedge T)-\omega(\pi_{k-1}^n\wedge T)
  \right)^2
  \delta_{\omega(\pi_{k-1}^n\wedge T)}
$$
($\delta_t$ being the Dirac measure at $t$)
on $[0,T]$ weakly converges to a finite measure $\xi$ on $[0,T]$
whose discrete part is given by the squared jumps of $\omega$:
\begin{equation}\label{eq:V-Follmer}
  V(t)
  =
  V\cont(t)
  +
  \sum_{s\in(0,t]}(\Delta\omega(s))^2,
\end{equation}
where $V$ is the distribution function of $\xi$
and $V\cont$ is its continuous part.
We will say that $V$ is \emph{F\"ollmer's quadratic variation
of $\omega$ along $\pi$}.

F\"ollmer's \cite{Follmer:1981} original definition of quadratic variation
is different from the definition above in the following respects:
it does not require the sequence $\pi$ of partitions to be nested;
it does not require $\pi$ to exhaust $\omega$;
it requires $\pi$ to be dense;
it assumes the unbounded time interval $[0,\infty)$.
F\"ollmer uses the notation $[\omega,\omega]$
for the function $V$ and does not use the term ``quadratic variation''
in respect of $V$.

We say that $\omega$ \emph{has Norvai{\u s}a's quadratic variation along $\pi$}
if there exists a function $V\in D[0,T]$ such that, for all $0\le s<t\le T$,
\begin{equation}\label{eq:Norvaisa}
  V(t)-V(s)
  =
  \lim_{n\to\infty}
  \sum_{k=1}^{\infty}
  \left(
    \omega(s\vee\pi_{k}^n\wedge t)-\omega(s\vee\pi_{k-1}^n\wedge t)
  \right)^2
\end{equation}
and, for any $t\in(0,T]$,
\begin{equation}\label{eq:V-Norvaisa}
  \Delta V(t)
  =
  (\Delta\omega(t))^2.
\end{equation}
The function $V$ is called \emph{Norvai{\u s}a's quadratic variation
of $\omega$ along $\pi$}.

Norvai{\u s}a's (\cite{Norvaisa:2001}, p.~1) original definition
is different from the definition in the previous paragraph
in the following respects:
it does not require $\pi$ to exhaust $\omega$;
it requires $\pi$ to be dense;
it only requires $\omega$ to be a regulated, rather than c\`adl\`ag, function.
Norvai{\u s}a calls $V$ the bracket function of $\omega$.

We can also weaken Norvai{\u s}a's requirements to the function $V$.
Namely, we will say that $\omega$ \emph{has weak quadratic variation along $\pi$}
if there exists a function $V\in D[0,T]$ such that
\begin{equation}\label{eq:weak}
  V(t)
  =
  \lim_{n\to\infty}
  \sum_{k=1}^{\infty}
  \left(
    \omega(\pi_{k}^n\wedge t)-\omega(\pi_{k-1}^n\wedge t)
  \right)^2
\end{equation}
at all points $t\in[0,T]$ of continuity of $V$
and (\ref{eq:V-Norvaisa}) holds for all $t\in(0,T]$.
The function $V$ is then called the \emph{weak quadratic variation of $\omega$ along $\pi$}.

The last two notions of quadratic variation are equivalent
to the one defined in Section~\ref{sec:QV-def}:

\begin{proposition}\label{prop:equivalence}
  Let $\omega\in D[0,T]$
  and $\pi$ be a nested sequence of partitions that exhausts $\omega$.
  The following three conditions are equivalent:
  \begin{enumerate}
  \item\label{it:my}
    the quadratic variation of $\omega$ along $\pi$ exists
    (in the sense of Section~\ref{sec:QV-def})
    and is $V=A^{\pi}$;
  \item\label{it:Norvaisa}
    Norvai{\u s}a's quadratic variation of $\omega$ along $\pi$ exists
    and is $V$;
  \item\label{it:weak}
    the weak quadratic variation of $\omega$ along $\pi$ exists
    and is $V$.
  \end{enumerate}
\end{proposition}

\ifJOURNAL
  \begin{pf}
\fi
\ifnotJOURNAL
  \begin{proof}
\fi
  First we assume condition~\ref{it:my} and prove condition~\ref{it:Norvaisa}.
  We know that (\ref{eq:Norvaisa}) (with $V:=A^{\pi}$) holds for $s=0$,
  so we assume $s>0$.
  In this case,
  (\ref{eq:Norvaisa}) is equivalent to
  \begin{multline}    
    \label{eq:to}
    \left(
      \omega(\pi^n_{\overline{k}(s,n)})
      -
      \omega(\pi^n_{\underline{k}(s,n)})
    \right)^2
    -
    \left(
      \omega(s)
      -
      \omega(\pi^n_{\underline{k}(s,n)})
    \right)^2
    \\   
    -
    \left(
      \omega(\pi^n_{\overline{k}(s,n)})
      -
      \omega(s)
    \right)^2
    \to
    0
    \quad
    (n\to\infty),
  \end{multline}    
  where $\underline{k}(s,n)$ is the largest $k$ satisfying $\pi^n_k\le s$
  and $\overline{k}(s,n)$ is the smallest $k$ satisfying $\pi^n_k\ge s$
  ((\ref{eq:to}) assumes that $\omega$ is not constant over $[s,t]$;
  the simple case where it is has to be considered separately).
  We can rewrite (\ref{eq:to}) as
  \begin{equation*}
    \left(
      \omega(\pi^n_{\overline{k}(s,n)})
      -
      \omega(s)
    \right)
    \left(
      \omega(s)
      -
      \omega(\pi^n_{\underline{k}(s,n)})
    \right)
    \to
    0,
  \end{equation*}
  which immediately follows from Lemma~\ref{lem:osc}:
  if $s\in(\underline{k}(s,n),\overline{k}(s,n))$ for all $n$
  (this is the non-trivial case),
  the first factor stays bounded and the second tends to 0 as $n\to\infty$.
  This proves (\ref{eq:Norvaisa}),
  and (\ref{eq:V-Norvaisa}) holds by Lemma~\ref{lem:increasing}.

  It is obvious that \ref{it:Norvaisa} implies \ref{it:weak}.

  Finally, we assume condition~\ref{it:weak} and prove \ref{it:my}.
  We know that $A^{n,\pi}\to V$ at all points of continuity of $V$
  as $n\to\infty$,
  and our goal is to prove that $A^{n,\pi}\to V$ uniformly.
  By the compactness of $[0,T]$,
  it suffices to prove that, for each point $t^*\in[0,T]$,
  $A^{n,\pi}\to V$ uniformly in some neighbourhood of $t^*$.
  Fix such $t^*$.
  We consider two cases separately:
  \begin{itemize}
  \item
    Suppose $\Delta V(t^*)=0$.
    Let $\epsilon\in(0,1)$.
    Choose $t'<t^*$ and $t''>t^*$ such that $\Delta V(t')=0$, $\Delta V(t'')=0$,
    and $V(t)$ belongs to $(V(t^*)-\epsilon,V(t^*)+\epsilon)$
    for all $t\in[t',t'']$.
    From some $n$ on, $A^{n,\pi}_{t'}$
    belongs to $(V(t')-\epsilon,V(t')+\epsilon)$
    and $A^{n,\pi}_{t''}$
    belongs to $(V(t'')-\epsilon,V(t'')+\epsilon)$.
    Therefore, from some $n$ on,
    both $A^{n,\pi}_{t'}$ and $A^{n,\pi}_{t''}$
    belong to $(V(t^*)-2\epsilon,V(t^*)+2\epsilon)$.
    From the proof of Lemma~\ref{lem:increasing}
    we know that, from some $n$ on, $A^{n,\pi}_t$ never drops by more than $\epsilon$
    as $t$ increases
    (see (\ref{eq:strange}), which we show to be impossible from some $n$ on
    when $t_1<t_2$; the simple argument in that proof only depends on conditions
    that are satisfied in our current context).
    Therefore, from some $n$ on,
    $A^{n,\pi}_{t}$ belongs to $(V(t^*)-3\epsilon,V(t^*)+3\epsilon)$
    for all $t\in[t',t'']$.
    We can see that, from some $n$ on,
    $A^{n,\pi}_{t}$ belongs to $(V(t)-4\epsilon,V(t)+4\epsilon)$
    for all $t\in[t',t'']$.
  \item
    Suppose that $\Delta V(t^*)\ne0$.
    Let $\epsilon\in(0,1)$.
    Choose $t'<t^*$ and $t''>t^*$ such that $\Delta V(t')=0$, $\Delta V(t'')=0$,
    $V(t)$ belongs to the interval $(V(t^*-)-\epsilon,V(t^*-)+\epsilon)$
    for all $t\in[t',t^*)$,
    and $V(t)$ belongs to $(V(t^*)-\epsilon,V(t^*)+\epsilon)$
    for all $t\in[t^*,t'']$.
    From some $n$ on,
    \begin{align}
      A^{n,\pi}_{t'}
      &\in
      \bigl(
        V(t^*-)-2\epsilon,V(t^*-)+2\epsilon
      \bigr),
      \label{eq:difficult-1a}\\
      A^{n,\pi}_{t''}
      &\in
      \bigl(
        V(t^*)-2\epsilon,V(t^*)+2\epsilon
      \bigr).
      \label{eq:difficult-1b}
    \end{align}
    From some $n$ on there is a $k$ such that $t^*=\pi^n_k$
    and $\pi^n_{k-1}\in(t',t^*)$
    (unless $\omega$ is constant over $(t',t^*)$,
    in which case $\pi^n_{k-1}$ should be replaced by $t'$ in the rest of this proof);
    and from some $n$ on, $\osc_{\pi^n}(\omega)<\epsilon$
    (Lemma~\ref{lem:osc}).
    Since $\Delta V(t^*)=(\Delta\omega(t^*))^2$,
    for such $n$ we have
    \begin{equation}\label{eq:difficult-2}
      A^{n,\pi}_{\pi^n_k}
      -
      A^{n,\pi}_{\pi^n_{k-1}}
      =
      \left(
        \omega(t^*) - \omega(\pi^n_{k-1})
      \right)^2
      >
      \Delta V(t^*) - 2\Delta_{\omega}\epsilon,
    \end{equation}
    where $\Delta_{\omega}:=\sup_{t\in(0,T]}\lvert\Delta\omega(t)\rvert$.
    We know that, from some $n$ on,
    $A^{n,\pi}_t$ never drops by more than $\epsilon$ as $t$ increases.
    According to (\ref{eq:difficult-1a}), (\ref{eq:difficult-1b}), and (\ref{eq:difficult-2}),
    we then have
    \begin{equation*}
      \begin{cases}
        A^{n,\pi}_{\pi^n_k}
        -
        A^{n,\pi}_{\pi^n_{k-1}}
        >
        \Delta V(t^*) - 2\Delta_{\omega}\epsilon\\
        A^{n,\pi}_{\pi^n_k}
	<
	V(t^*) + 3\epsilon\\
        A^{n,\pi}_{\pi^n_{k-1}}
        >
	V(t^*-) - 3\epsilon.
      \end{cases}
    \end{equation*}
    This system of three inequalities immediately implies
    \begin{equation*} 
      \begin{cases}
        A^{n,\pi}_{\pi^n_k}
	>
	V(t^*) - 2\Delta_{\omega}\epsilon - 3\epsilon\\
        A^{n,\pi}_{\pi^n_{k-1}}
	<
        V(t^*-) + 2\Delta_{\omega}\epsilon + 3\epsilon.
      \end{cases}
    \end{equation*}
    Combining this with (\ref{eq:difficult-1a}), (\ref{eq:difficult-1b}),
    and $\osc_{\pi^n}(\omega)<\epsilon$,
    we can see that
    $$
      \begin{cases}
        A^{n,\pi}_t\in(V(t^*)-2\Delta_{\omega}\epsilon-4\epsilon,V(t^*)+3\epsilon)
	  & \text{if $t\in[t^*,t'']$}\\
        A^{n,\pi}_t\in(V(t^*-)-3\epsilon,V(t^*-)+2\Delta_{\omega}\epsilon+4\epsilon)
	  & \text{if $t\in[t',t^*)$}.
      \end{cases}
    $$
    Therefore, from some $n$ on, we have
    $$
      \left|
        A^{n,\pi}_t - V(t)
      \right|
      <
      2\Delta_{\omega}\epsilon + 5\epsilon
    $$
    for all $t\in[t',t'']$.
  \end{itemize}
  In both cases $A^{n,\pi}$ converges to $V$ uniformly
  in some neighbourhood of $t^*$,
  which completes the proof.
\ifJOURNAL
  \qed
  \end{pf}
\fi
\ifnotJOURNAL
  \end{proof}
\fi

On the other hand, F\"ollmer's notion of quadratic variation is different
and even anomalous unless $\pi$ is dense.
Set, e.g.,
$$
  \omega(t)
  :=
  \begin{cases}
    0 & \text{if $t\in[0,T/2)$}\\
    1 & \text{if $t\in[T/2,T)$}\\
    0 & \text{if $t=T$}
  \end{cases}
$$
and consider the sequence of nested partitions
$$
  \pi^1=\pi^2=\cdots
  :=
  (0,T/2,T,\infty,\infty,\ldots),
$$
which exhausts $\omega$.
The three definitions whose equivalence is asserted in Proposition~\ref{prop:equivalence}
give the same quadratic variation $V$ of $\omega$ along $\pi$,
$$
  V(t)
  :=
  \begin{cases}
    0 & \text{if $t\in[0,T/2)$}\\
    1 & \text{if $t\in[T/2,T)$}\\
    2 & \text{if $t=T$},
  \end{cases}
$$
whereas F\"ollmer's does not exist:
the first part of the definition gives
$$
  V(t)
  :=
  \begin{cases}
    1 & \text{if $t\in[0,T/2)$}\\
    2 & \text{if $t\in[T/2,T]$},
  \end{cases}
$$
which fails to satisfy (\ref{eq:V-Follmer}) at $t=T$
($V(0)\ne0$ also looks anomalous).
The anomalies disappear when $\pi$ is dense:

\begin{proposition}\label{prop:non-equivalence}
  Let $\omega\in D[0,T]$
  and $\pi$ be a dense nested sequence of partitions.
  \begin{enumerate}
  \item\label{it:direction1}
    If the quadratic variation $V$ of $\omega$ along $\pi$ exists
    in the sense of any of the three definitions of Proposition~\ref{prop:equivalence},
    $V$ is also F\"ollmer's quadratic variation of $\omega$ along $\pi$.
  \item\label{it:direction2}
    If F\"ollmer's quadratic variation $V$ of $\omega$ along $\pi$ exists,
    $V$ is also the quadratic variation of $\omega$ along $\pi$
    in the sense of the three definitions of Proposition~\ref{prop:non-equivalence}.
  \end{enumerate}
\end{proposition}

\ifJOURNAL
  \begin{pf}
\fi
\ifnotJOURNAL
  \begin{proof}
\fi
  We start from \ref{it:direction1}.
  Let $V$ be the quadratic variation of $\omega$ along $\pi$
  in the sense of the definitions of Proposition~\ref{prop:equivalence}.
  Since (\ref{eq:V-Follmer}) is obviously equivalent to (\ref{eq:V-Norvaisa}),
  the definition of F\"ollmer's quadratic variation
  shows that it suffices to prove
  \begin{equation}\label{eq:Follmer}
    \sum_{k\in\bbbn:\pi^n_{k-1}\le t}
    \left(
      \omega(\pi_{k}^n\wedge T)-\omega(\pi_{k-1}^n)
    \right)^2
    \to
    V(t)
  \end{equation}
  as $n\to\infty$,
  where $t\in[0,T]$ is such that $V$ is continuous at $t$,
  i.e., by (\ref{eq:V-Norvaisa}), $\omega$ is continuous at $t$.
  Comparing (\ref{eq:Follmer}) with (\ref{eq:weak}),
  we can see that, furthermore, it suffices to prove
  \begin{equation}\label{eq:remainder}
    \left(
      \omega(t)
      -
      \omega(\pi^n_{\underline{k}(t,n)})
    \right)^2
    -
    \left(
      \omega(\pi^n_{\underline{k}(t,n)+1}\wedge T)
      -
      \omega(\pi^n_{\underline{k}(t,n)})
    \right)^2
    \to
    0,
  \end{equation}
  which immediately follows from the continuity of $\omega$ at $t$
  provided $\pi$ is dense.

  The argument in the previous paragraph
  (cf.\ (\ref{eq:remainder}))
  also establishes \ref{it:direction2}.
\ifJOURNAL
  \qed
  \end{pf}
\fi
\ifnotJOURNAL
  \end{proof}
\fi

\begin{remark}
  F\"ollmer's \cite{Follmer:1981} and Norvai{\u s}a's \cite{Norvaisa:2001}
  condition that the sequence of partitions $\pi$ should be dense
  would not in fact be a big obstacle in this paper:
  to make the random sequence of partitions $\tau^n$
  formed by the stopping times $\tau^n_k$
  (as defined in Section~\ref{sec:QV-proof} or Section~\ref{sec:multidim})
  dense
  we can simply complement $\tau^n(\omega)$
  by the points $T k 2^{-n}$, $k=1,\ldots,2^{n}-1$.
  The properties \ref{it:below}--\ref{it:above} of Section~\ref{sec:multidim}
  will be still satisfied after this extension and the sequence will be still nested.
  This step is, however, awkward, and we avoid it.
\end{remark}

\section{Implications of the existence of quadratic variation}
\label{sec:implications}

This section reviews some known implications
(the simpler ones from \cite{Follmer:1981} and \cite{Norvaisa:2001})
of the existence of quadratic variation.
Let $\omega\in D[0,T]$ and $\pi=(\pi^n)$ be a nested sequence of partitions
that exhausts $\omega$.
We are mainly interested in the case
where $\pi=(\pi^n(\omega))$ is the sequence of partitions
formed by the stopping times $\tau^n_k$,
as defined in Section~\ref{sec:QV-proof}.
We know that in this case the quadratic variation of $\omega$ exists
unless $\omega$ is in a null set (Theorem~\ref{thm:1D}).
To simplify notation,
we do not consider the multidimensional case
(Section~\ref{sec:multidim} and Theorem~\ref{thm:multidim}).

Throughout this section, we fix $\omega\in D[0,T]$ and $\pi=(\pi^n)$
such that $\pi$ is a nested sequence of partitions that exhausts $\omega$
and the quadratic variation of $\omega$ along $\pi$ exists
(cf.\ Proposition~\ref{prop:equivalence}).
As in Section~\ref{sec:multidim},
we use square brackets to denote the quadratic variation along $\pi$,
when it exists;
e.g., $[\omega]$ is the quadratic variation of $\omega$ along $\pi$.
A minor difference of the results that we state in this section
from the original ones stated in \cite{Follmer:1981} and \cite{Norvaisa:2001}
is that we do not assume that $\pi$ is dense:
the assumption of denseness becomes redundant under our assumptions
since we can always remove all intervals of constancy
from the domain of $\omega$.
\Extra{This might lead to repeated entries in $\pi^n$, of course.}

\subsection*{Stochastic integration and It\^o's lemma}

Suppose $f:\bbbr\to\bbbr$.
For each $t\in[0,T]$ and $\omega\in D[0,T]$,
define
\begin{equation}\label{eq:stochastic-integral}
  \int_0^t f(\omega(s-)) \dd\omega(s)
  :=
  \lim_{n\to\infty}
  \sum_{k=0}^{\infty}
  f(\omega(\pi^n_k(\omega)\wedge t))
  \left(
    \omega(\pi^n_{k+1}\wedge t)
    -
    \omega(\pi^n_{k}\wedge t)
  \right).
\end{equation}
The following two propositions are F\"ollmer's (\cite{Follmer:1981}, p.~144) theorem
adapted to our framework.
The first proposition says that the stochastic integral (\ref{eq:stochastic-integral}) exists.
\begin{proposition}\label{prop:1}
  Let $f\in C^1(\bbbr)$.
  The limit in (\ref{eq:stochastic-integral}) exists for each $t\in[0,T]$,
  and as a function of $t\in[0,T]$ is an element of $D[0,T]$.
\end{proposition}
The second proposition gives the It\^o--F\"ollmer formula.
\begin{proposition}\label{prop:2}
  Let $F\in C^2(\bbbr)$.
  For all $t\in[0,T]$,
  \begin{align*}
    F(\omega(t))
    &=
    F(\omega(0))
    +
    \int_0^t
    F'(\omega(s-))\dd\omega(s)
    +
    \frac12
    \int_0^t
    F''(\omega(s))\dd[\omega]\cont_s\\
    &\quad+
    \sum_{s\in(0,t]}
    \Bigl(
      \Delta F(\omega(s))
      -
      F'(\omega(s-))\Delta\omega(s)
    \Bigr).
  \end{align*}
\end{proposition}

For the extension of the stochastic integral and the It\^o--F\"ollmer formula
to the case of $\omega$ taking values in $\bbbr^m$,
see \cite{Follmer:1981}, pp.~147--148.
This extension uses the quadratic covariation processes:
see (\ref{eq:covariation}).

An important
development of F\"ollmer's results
is their extension by Cont and Fourni\'e \cite{Cont/Fournie:2010}
to non-anticipative functionals.
The existence of the limit in (\ref{eq:stochastic-integral})
when $f(\omega(\pi^n_k(\omega)\wedge t))$ is replaced by
$g(\omega|_{[0,\pi^n_k(\omega)\wedge t]})$,
where $g$ is a functional satisfying certain regularity conditions
(including being non-anticipative),
is established in Theorem~4 of \cite{Cont/Fournie:2010}.
The same theorem gives an It\^o--F\"ollmer formula
for non-anticipative functionals.

\subsection*{Quadratic variation for other processes}


We now state some known results
for the existence of quadratic variation
for two kinds of processes different from the basic process $W_t(\omega):=\omega(t)$:
namely, for functions of the form $f(\omega(t))$, where $f\in C^1(\bbbr)$,
and for stochastic integrals w.r.\ to $\omega$.

Suppose $f\in C^1(\bbbr)$.
F\"ollmer notices in \cite{Follmer:1981} (Remark 2 on p.~148)
that a standard argument in the theory of stochastic integration
(as in \cite{Meyer:1976}, Theorem VI.5 on p.~359)
implies that the quadratic variation of $f(\omega)$ along $\pi$ exists
and is equal to
$$
  [f(\omega)]_t
  =
  \int_0^t
  \left(
    f'(\omega(s))
  \right)^2
  \dd[\omega]\cont_s
  +
  \sum_{s\in(0,t]}
  \bigl(
    \Delta f(\omega(s))
  \bigr)^2.
$$
Norvai{\u s}a's Theorem 3.26 in \cite{Norvaisa:2001}
implies that the quadratic variation of the stochastic integral
$
  \Phi(t)
  :=
  \int_0^t
  f(\omega(s-))
  \dd\omega(s)
$
exists and is equal to
$$
  [\Phi]_t
  =
  \int_0^t
  f^2(\omega(s-))
  \dd[\omega]_s.
$$

Suppose that $\omega\in D[0,T]$ is positive
(cf.\ the second example discussed in Section~\ref{sec:definitions}).
In this case it is natural to measure quadratic variation
on the relative rather than absolute scale,
and so to consider the quadratic variation of the logarithm of $\omega$.
By F\"ollmer's result,
the quadratic variation process of $\ln\omega$ is
$$
  [\ln\omega]_t
  =
  \int_0^t
  \frac{\dd[\omega]\cont_s}{\omega^2(s)}
  +
  \sum_{s\in(0,t]}
  \bigl(
    \Delta\ln\omega(s)
  \bigr)^2.
$$
By Norvai{\u s}a's result,
the quadratic variation process of the stochastic logarithm
$
  \Ln\omega(t)
  :=
  \int_0^t
  \frac{\dd\omega(s)}{\omega(s-)}
$
of $\omega$ is
$$
  [\Ln\omega]_t
  =
  \int_0^t
  \frac{\dd[\omega]_s}{\omega^2(s-)}.
$$

\section{Conclusion}

In this section we discuss possible directions of further research.
This paper shows the existence of quadratic variation for typical price paths
in $\Omega_{\psi}$.
It is easy to see that Theorem~\ref{thm:1D} becomes false
if we simply set $\psi:=\infty$,
but an interesting question is whether we can set $\psi:=\infty$
if we only consider nonnegative $\omega$
(for many securities, $\omega\ge0$ can be assumed
from economic considerations).

Another possible way to get rid of the assumption that the jumps of $\omega$ are bounded
by a function of $\omega$
is to allow trading in American, or binary American, options
to hedge against huge jumps of $\omega$.
If the prices of such out-of-the-money options
tend to zero sufficiently fast as their moneyness decreases,
we can expect that the analogue of Theorem~\ref{thm:1D}
will continue to hold even for $\psi:=\infty$.

In Section~\ref{sec:QV-proof},
we proved the existence of quadratic variation
only for a specific array of stopping times $(\tau^n_k)$.
In Section~\ref{sec:multidim},
we noticed that the argument of Section~\ref{sec:QV-proof}
works for the class of arrays $(\tau^n_k)$ which we called arrays of dyadic type
and that any two arrays of dyadic type lead to the same values of quadratic variation
for typical $\omega$.
It is clear that this observation can be extended
to a much wider class of arrays.

One interpretation of the stochastic integral~(\ref{eq:stochastic-integral})
is that it is the capital of a trading strategy.
We, however, also define capital processes directly:
see (\ref{eq:simple-capital}) and (\ref{eq:nonnegative-capital}).
Can all nonnegative capital processes (\ref{eq:nonnegative-capital})
be represented as stochastic integrals?
It is clear that the ``Markovian'' definition~(\ref{eq:stochastic-integral})
is not sufficient
(the trading strategy in~(\ref{eq:stochastic-integral})
takes into account only the current price),
so this question is about the extension of~(\ref{eq:stochastic-integral})
to non-anticipating functionals, as in \cite{Cont/Fournie:2010}.

\subsection*{Acknowledgments}


An earlier version of this paper was the basis of my talk at the 2014 Vilnius Conference
on Probability Theory and Mathematical Statistics 
(section ``Random processes'', session ``Rough paths''),
and I am grateful to the organizers for inviting me and to the listeners for useful comments,
with special thanks to Rimas Norvai{\u s}a.
This research was supported by the Air Force Office of Scientific Research
(grant FA9550-14-1-0043).

\end{document}